\documentclass[11pt]{article}
\usepackage[margin=1in]{geometry}
\usepackage{amsmath, amsthm, amsfonts, amssymb, epsfig, subfig}
\usepackage{tikz}
\usepackage{subfig}
\usetikzlibrary{arrows,shapes}

\newcommand{\E}{\ensuremath{\mathbb E}}
\newcommand{\p}{\ensuremath{\mathbb P}}
\newcommand{\Z}{\ensuremath{\mathbb Z}}

\newcommand{\cM}{\mathcal{M}}

\newcommand{\R}{\mathbb{R}}
\newcommand{\cP}{\mathcal{P}}
\newcommand{\cV}{\mathcal{V}}

\newcommand{\union}{\cup}

\newcommand{\tri}{\oplus}

\newcommand{\ubar}{\overline{u}}
\newcommand{\xbar}{\overline{x}}

\newcommand{\pleq}{\preccurlyeq}

\def\bemph{\em}
\def\e{{\rm e}}
\newcommand{\m}{\cM}

\newcommand{\mnn}{\cM_{{nn}}}
\newcommand{\mt}{\cM_{{T}}}

\newcommand{\mcol}{\cM_{{col}}}

\usepackage{ifthen} 
\newboolean{includefigs} 
\setboolean{includefigs}{true} 
\newcommand{\condcomment}[2]{\ifthenelse{#1}{#2}{}}
\theoremstyle{plain}
\newtheorem{theorem}{Theorem}[section]

\newtheorem{definition}{Definition}[section]
\newtheorem{lemma}[theorem]{Lemma}

\newcommand{\mtil}{\cM_{{B}}}
\newcommand{\mmon}{\cM_{{B}}}
\newcommand{\mfluct}{\cM_{{B}}}

\makeatletter
\def\bigpar{\bigbreak\@afterindentfalse\@afterheading\ignorespaces}
\def\medpar{\medbreak\@afterindentfalse\@afterheading\ignorespaces}
\def\smallpar{\smallbreak\@afterindentfalse\@afterheading\ignorespaces}

\newlength{\saveindent}
      {\bigpar{\bf Proof.}\ %  previously \sentsp rather than \
             \setlength{\saveindent}{\parindent}
                       \ignorespaces}%
      {\stopproof\ignorespaces\bigbreak \setlength{\parindent}{\saveindent}}

\newenvironment{Proofof}[1]%
      {\bigpar{\bf Proof of #1.}\ %
             \setlength{\saveindent}{\parindent}
                       \ignorespaces}%
      {\stopproof\ignorespaces\bigbreak \setlength{\parindent}{\saveindent}}

\def\endex{\stopproof\end{example}}
\def\square{\vbox{\hrule height.2pt\hbox{\vrule width.2pt height5pt \kern5pt
                                   \vrule width.2pt} \hrule height.2pt}}
\def\stopproof{\qquad\square}

\newcommand{\comment}[1]{}

\title{Sampling Biased Monotonic Surfaces using Exponential Metrics\footnote{A preliminary version of this paper appeared in the Proceedings of the 20th ACM/SIAM Symposium on Discrete Algorithms, 76--85, 2009.}}
\author{Sam Greenberg\thanks{Department of Defense, Arlington, VA.}
\and Dana Randall\thanks{School of Computer Science, Georgia Institute of Technology,
Atlanta GA, 30332-0765. Supported in part by NSF grants CCF-1526900, CCF-1637031 and CCF-1733812.}
\and Amanda Pascoe Streib\thanks{Center for Computing Sciences, Bowie, MD 20715.}
}

\begin{document}
\date{}
\maketitle

\thispagestyle{empty}  
\begin{abstract}
Monotonic surfaces spanning finite regions of $\Z^d$ arise in many contexts, including
DNA-based self-assembly, card-shuffling and lozenge tilings.   One method that has been used to uniformly generate these surfaces 
is a Markov chain that iteratively adds or removes a single cube below the surface during a step.  
We consider a biased version of the chain, where we are more likely to add a cube than to remove it, thereby favoring surfaces that are ``higher'' or have more cubes below it.  We prove that the chain is rapidly mixing for any uniform bias in $\Z^2$ and for bias $\lambda > d$ in $\Z^d$ when $d>2$.  In $\Z^2$ we match the optimal mixing time achieved by Benjamini et al.  in the context of biased card shuffling \cite{BBHM05}, but using much simpler arguments.  The proofs use a geometric distance function and a variant of  path coupling in order to handle distances that can be exponentially large.   We also provide the first results in the case of \emph{fluctuating bias}, where the bias can vary depending on the location of the tile.   We show that the chain continues to be rapidly mixing if the biases are close to uniform, but that the chain can converge exponentially slowly in the general setting. 

\end{abstract}

\section{Introduction}\label{BRintro}
In this paper, we are concerned with designing provably efficient algorithms for sampling from a family of discrete monotonic surfaces, where we bias the distribution to favor surfaces that are ``higher.'' There is a long history of sampling from various families of monotonic surfaces because many natural combinatorial problems are known to have an associated height function that can be interpreted as a piecewise-linear surface.  
In statistical physics, for example, domino tilings of the Cartesian lattice and lozenge tilings of the triangular lattice are natural models of diatomic molecules that have associated 3-dimensional height functions. Similarly, states of the zero temperature Potts model, a popular model of ferro-magnetism, have a height function that maps 3-colorings  in $\Z^{d-1}$ to surfaces in $\Z^{d}$.
%3-colorings of lattice regions, representing states of the zero temperature Potts model, a popular model of ferro-magnetism, have a height function that maps 3-colorings  in $\Z^d$ to surfaces in $\Z^{d+1}$.  
In each case, random sampling provides insight into the thermodynamic properties of these physical systems. We focus on surfaces that are unions of $d-1$ dimensional faces of the Cartesian lattice $\Z^d$ and discuss extensions to these other natural families at the end of the paper.

In two dimensions, monotonic surfaces, called \emph{staircase walks}, are paths within a finite region of the lattice $\Z^2$ that step to the right and down at every edge (see Figure~\ref{MonSurf}a).  
Markov chains for sampling staircase walks have been used to analyze card-shuffling algorithms by associating to a permutation a set of staircase walks~\cite{wilson, BBHM05}.  One simple Markov chain $\cM_U$ for sampling uniformly from the set of staircase walks, known as the ``mountain / valley chain,'' tries to invert a mountain that moves to the right and then down to a valley that goes down and then to the right, or vice versa.  This chain has also been studied in the context of  \emph{Dyck paths}, or staircase walks that start at $(0,n)$, end at $(n, 0)$ and do not cross below the line $x+y=n$.  Dyck paths are enumerated by the $n$th Catalan number, and $\cM_U$ has proven useful for sampling from these and other Catalan structures \cite{MT98}.  Wilson \cite{wilson} gave tight bounds on the convergence rate of the mountain / valley chain in the general case and in the case of Dyck paths by showing that in both cases it mixes in time $\Theta(n^3 \log n)$.

\begin{figure}[ht]
\centering
\includegraphics[scale=.42]{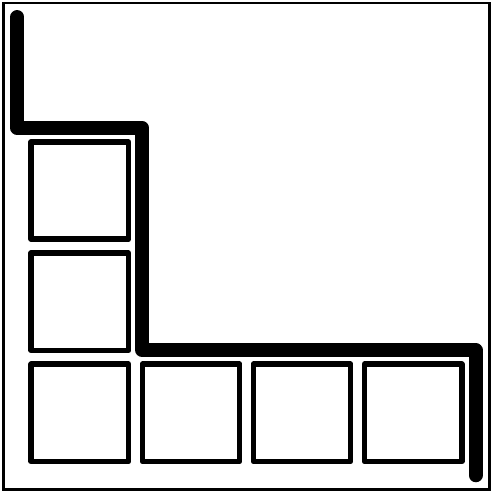}\put(-36,-12){(a)}
%\put(8,25){$\rightarrow$}
%\hspace{.3in}
%\includegraphics[scale=.42]{2d.pdf}
\hspace{1in}
\includegraphics[scale=.33]{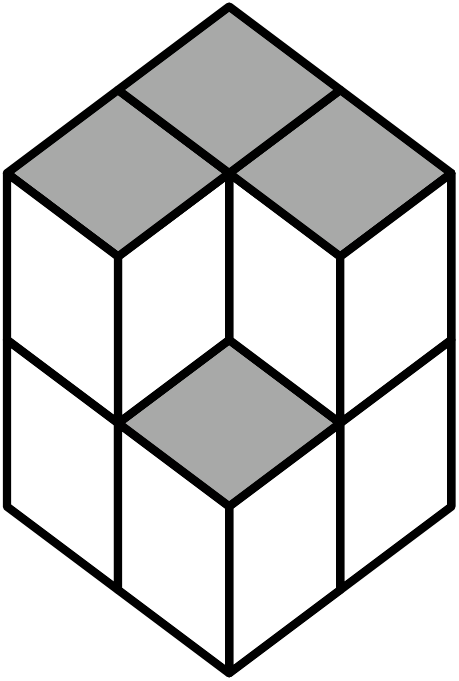}
%\put(8,25){$\rightarrow$}
\put(-29,-12){(b)} 
%\hspace{.3in}
%\includegraphics[scale=.33]{PathCoup1.pdf}
\caption{Monotonic surfaces in two and three dimensions.}
\label{MonSurf}
\end{figure}

Monotonic surfaces in ${\Z}^3$ correspond bijectively with {\it lozenge tilings} of finite regions of the triangular lattice \cite{lrs}.  A lozenge tiling is a covering of the lattice region with {\it lozenges}, or rhombuses that cover two adjacent triangular faces.  Thus, a lozenge tiling is just a perfect matching (or dimer covering) in the dual graph.  When we look at a two-dimensional picture of a lozenge tiling, our eyes automatically interpret the picture as a surface that is the upper envelope of a set of supported cubes in $\Z^3$ (see Figure~\ref{MonSurf}b).  We call this set of cubes $\sigma$ a \emph{downset}, since the set is downwardly-monotonic, and we let $M(\sigma)$ denote the monotonic surface formed by $\sigma$; clearly these surfaces and downsets are also in bijection.  
%Similar families of height functions and monotonic surfaces arise from other combinatorial structures, such as 3-colorings and domino tiling of finite regions in ${\Z}^2$.

A natural Markov chain for uniformly sampling lozenge tilings tries to identify three closely packed lozenges forming a hexagon and rotates them by 180 degrees.  Equivalently, we interpret this move quite naturally using the bijection with surfaces and this move corresponds to perturbing the surface locally by adding or removing a single cube, so this is precisely $\cM_U$ in $\Z^3$.  It is easy to see that this chain connects the state space because, starting from any configuration, we can remove cubes until we reach the ``empty configuration.''  This Markov chain is known to be rapidly mixing, or quickly converging to equilibrium, so it is effective for efficiently generating samples from close to the uniform distribution \cite{lrs, rt, wilson, cmt}.
%Local Markov chains that update a single site of the coloring, or rotate 2 or 3 tiles in a tiling, can be used to generate random configurations efficiently.   Luby, Randall and Sinclair \cite{lrs}  and Wilson \cite{wilson} analyzed a family of related, but nonlocal, Markov chains on lozenge tilings and 3-colorings and showed that they are rapidly mixing (or converging rapidly to equilibrium).
%Subsequently, Randall and Tetali \cite{rt} showed that bounds on the mixing times of the nonlocal chains imply that the local chains are also rapidly mixing.  Finally, Caputo et al. \cite{cut} gave optimal bounds of $\Theta(n^3 \log n)$ %%%% check
%for the local chain that adds and removes individual cubes based on an approach bounding the mean curvature of the surface.
%  cmt:    http://arxiv.org/abs/1101.4190

There also  has been interest in a {\it biased} version $\mmon$ of this local Markov chain, where we are 
more likely to add unit cubes than remove them.   More precisely, if $\sigma$ is formed from~$\tau$ by the addition of a single cube at
position $\bar{x}$, then  $\lambda_{\bar{x}} = P(\sigma, \tau)/P(\tau, \sigma)$ is the {\it bias at $\bar{x}$}.  If~$\lambda_{\bar{x}} > 1$ for every $\bar{x}$, then the stationary distribution favors configurations with more cubes.

Biased surfaces arise in nanoscience in the context of DNA-based self-assembly growth models (see, e.g., \cite{fs, seeman, winfree, wys}).  In this setting, roughly ``square'' shaped tiles are constructed from strands of DNA so that each side of the tiles is single-stranded. Certain pairs of  tiles are encouraged to line up and attach along edges by encoding corresponding sides with complementary sequences of base pairs. At appropriately chosen temperatures, these tiles will have a good chance of assembling according to these prescribed rules, although they also have a chance of disassociating and breaking apart.  Majumder et al.~\cite{reif+} consider a DNA self-assembly model that allows the left column and bottom row of a large square to first form, and then iteratively allows tiles to associate with the large substrate if their left and bottom neighbors are already present (see Figure~\ref{MonSurf}a). Likewise, tiles can disassociate if their upper and right neighbors are not present, although disassociation happens at a lower rate. The dynamics of this model are precisely captured by the local Markov chain $\mmon$ on 2-dimensional monotonic surfaces and the chain must be rapidly mixing if the substrate is to efficiently self-assemble, as required. The 3-dimensional analogue has also been used to study self-assembly, where now tiles are shaped like cubes (as in Figure~\ref{MonSurf}b) and complementary sequences are used to encourage specified pairs of faces to attach.  
The problem of generating biased surfaces in two-dimensions was  independently studied in the context of biased card shuffling, where we allow nearest-neighbor transpositions but favor putting each pair in order at each step \cite{BBHM05}.

Previous work has focused primarily on the case where the biases are {\it uniform}; that is, $\lambda_{\bar{x}}=\lambda$ for every $\bar{x}$, for which the stationary probability will be proportional to $\lambda^{|\sigma|}$, where $|\sigma|$ is the number of unit cubes lying below the surface $\sigma$.  
In two dimensions, the uniform bias Markov chain is equivalent to an asymmetric exclusion process, which Benjamini et al. \cite{BBHM05} studied in order to analyze a biased card shuffling algorithm that favors putting each pair of cards in the lexicographically correct order.  They give a bound of $\Theta(n^2)$ on the mixing rate of the biased chain on $h \times w$ regions of $\Z^2$ (where $h+w=n$) for any uniform bias $\lambda>1$ that is a constant.  The bounds are optimal when $h=w=n/2$.  For the three dimensional variant, Caputo et al.~\cite{cmt} recently proved that the biased chain mixes in time $\widetilde{O}(n^3)$ for any constant bias $\lambda$, where the $\widetilde{O}$ notation suppresses logarithmic factors.  
In dimensions $d>3$ almost nothing is known, in either the biased or unbiased settings. 
%Majumder et al.~\cite{reif+} show that the chain mixes quickly when the uniform bias is $\Omega(n)$, as the case of large bias is the most interesting for nanotechnology applications. 
%Nothing else is known about the convergence of the biased chain;  both of these results are highly technical and do not readily generalize to other values of the bias or other dimensions.
%%%\marginpar{DR:  check}

\subsection{Our results}\label{results}
We make progress in several aspects of the problem of sampling biased surfaces. 
In two dimensions we show that the biased chain is rapidly mixing for any uniform bias on a large family of simply-connected regions, even when the bias is arbitrarily close to one.  Our proof is significantly simpler than the arguments of Benjamini et al., while achieving the same optimal bounds on the mixing time for square regions when the bias is constant. In fact, on rectangular $h\times w$ regions of $\Z^2$, where $h\leq w$, we get improved bounds of $O(w(h+\ln w))$, which is optimal when $h=\Omega(\ln w)$.
We also show the chain is rapidly mixing on $d$-dimensional lattice regions provided the bias $\lambda \geq d^2$. Again, our bounds on the mixing time are optimal when the regions are hyper-cubes, and show the chain is rapidly mixing for a large family of  simply-connected regions.  
The key observation underlying these results is that there is an exponential metric on the state space such that the distance between pairs of configurations is always decreasing in expectation. 
We then show how to modify the path coupling theorem to handle the case when the distances are exponentially large and the expected change in distance is small during moves of the coupled chain.  Previously, Berger et al.~\cite{bkmp} also used an exponential metric in the context of Glauber dynamics on trees, appealing to the multiplicative version of the path coupling theorem originally given by Bubley and Dyer~\cite{BubDyer}.  Our version of the path coupling theorem makes explicit when we can use exponential metrics to bound convergence times.
We believe that this new theorem is of independent interest, and it has already been used, for example,  in the context of sampling lattice triangulations \cite{cmss} and rectangular dissections \cite{cmr}.

Last, we consider ``fluctuating bias,'' where the rate at which we add or remove a cube depends on its location.  This setting is quite natural for the self-assembly growth process where the tile at a particular location may have site-specific sequences along its bounding edges.  We show that the fast convergence results still hold whenever the biases are close to uniform over locations.  In fact, in two dimensions, we show that the chain mixes in time $O(n^2)$ even if the biases are not close together, as long as all of the biases are bounded away from one by a constant.  However, in the general setting we may see very different behavior.  We construct  an example where the convergence rate requires exponential time starting at any initial configuration.  In this example every move occurs with at least inverse polynomial probability and at each location it is at least as likely to add a cube as to remove.  This demonstrates that the behavior of these growth processes is quite complicated in the case of fluctuating bias.

The remainder of the paper is organized as follows. In Section~\ref{Path Coupling} we review the path coupling method and introduce the modified path coupling theorem that is more appropriate when distances are exponentially large. In Section~\ref{UnifSection}, we formalize the model and Markov chain and show how to bound the mixing time of the chain.  In Section~\ref{FluctSection}, we generalize these techniques to apply to the setting of fluctuating bias.  %, and exhibit a set of biases such that $\lambda_{\bar{x}}> 1$ for all $\bar{x}\in R$ and yet the Markov chain takes exponential time to mix.  
Finally, in the last section we discuss other related problems, including sampling biased 3-colorings in ${\Z}^d$.

\section{The Markov chain}\label{BRmodel}
\label{tilSection}
%The nanoscience application that motivates this work is a model of DNA-based self-assembly (see, e.g., \cite{fs, seeman, winfree, wys}).  Roughly ``square'' shaped tiles are constructed from strands of DNA so that each side of the tiles is single-stranded. Certain pairs of  tiles are encouraged to line up and attach along edges by encoding corresponding sides with complementary sequences of base pairs. At appropriately chosen temperatures, these tiles will have a good chance of assembling according to these prescribed rules, although they also have a chance of disassociating and breaking apart. The model considered by Majumder et al.~\cite{reif+} allows the left column and bottom row of a large square to first form, and then iteratively allows tiles to associate with the large substrate if their left and bottom neighbors are already present (see Figure~\ref{MonSurf}a). Likewise, tiles can disassociate if their upper and right neighbors are not present, although disassociation happens at a lower rate. The dynamics of this model are precisely captured by the local Markov chain $\mmon$ on 2-dimensional monotonic surfaces and the chain must be rapidly mixing if the substrate is to efficiently self-assemble, as required. The 3-dimensional analogue is also used to study self-assembly, where now tiles are shaped like cubes (as in Figure~\ref{MonSurf}b) and complementary sequences are used to encourage corresponding faces to attach.

We can now formalize our definition of monotonic surfaces and the  Markov chain that makes local updates to these surfaces. Throughout the majority of the paper we focus on surfaces arising in the context of staircase walks, lozenge tilings, and higher dimensional surfaces that bound downsets of unit boxes in $\Z^d$.  We discuss the generalizations to other families of monotone surfaces such as those arising from 3-colorings in $\Z^d$ in the conclusions.

%\subsection{Monotonic surfaces}\label{monotonicsurfaces}
We start by  first considering monotonic surfaces forming over simple, rectangular regions $R$ in $\Z^2$.  Later we will show that we can generalize these results to a family of simply-connected regions in $\Z^d$ that we call {\it nice}.  
The generalization is straightforward, but requires some careful notation, so for simplicity we postpone the details until Section~\ref{generalization}. 
Recall that in $\Z^2$, a  \emph{monotonic surface (or path)} in $R$ is a path starting and ending on the boundary of $R$ that only takes steps down or to the right and is composed entirely of edges with both endpoints in $R$.
Such a path is illustrated in Figure \ref{MonSurf}a when ${R}$ is a $4 \times 4$ square.  Notice that any monotonic surface can be interpreted as the upper boundary of a set of unit squares (which we call a downset), where each square in the set is supported below or to the left by other squares in the set or the boundary of $R$.   
Let $M$ be the bijection between downsets $\sigma$ and their corresponding monotonic surfaces $M(\sigma)$. 
We let the state space $\Omega_{mon}$ be the set of all downsets of $R$.

%\subsection{The biased Markov chain}

Now we can describe the Markov chain $\mmon$ on $\Omega_{mon}$.   For simplicity, 
we start by defining the {\it unbiased} chain $\cM_U$ that converges to the {\it uniform distribution} over monotonic paths 
$\Omega_{mon}$.
Start at an arbitrary downset, e.g., let $\sigma_0=R_L$, where $R_L$ is the empty downset, and repeat the following steps. If we are at a downset $\sigma_t$ at time $t$, pick a diagonal $d$ that is parallel to the vector $\ubar^*=(1,1)$ and which intersects the monotonic path at a vertex $v$.  Also, pick an integer $b \in \pm 1$ uniformly at random. If $b=+1$, add the cube above and to the right of the vertex $v$ to create $\sigma_{t+1}$, if this is a valid downset. If $b=-1$, let $\sigma_{t+1}$ be obtained from $\sigma_t$ by removing the cube below and to the left of $v$ if this is a valid downset.
In all other cases, keep $\sigma_t$ unchanged so that $\sigma_{t+1} = \sigma_t$.  

\begin{lemma}\label{conlemma} For any rectangular region $R$, the Markov chain $\cM_U$  connects the state space
$\Omega_{mon}$. 
\end{lemma}

\begin{proof}
Let $\sigma$ be any downset and let $x^{max}$ be any cube
in $\sigma$ such that $\sum_i x_i^{max}$ is maximized, if it exists.  We can always remove $x^{max}$
and move to $\sigma' = \sigma \setminus x^{max}$ without violating the downset condition. Thus,
from any valid downset $\sigma$ we can always remove points and get to the ``lowest'' downset~$R_L$. Also, such a sequence of steps can be reversed to move from $R_L$ to any other downset~$\rho$.
\end{proof}

\noindent Since we have shown that the moves of $\cM_U$ connect the state space and all valid moves have the same transition
probabilities, we can conclude from detailed balance that the chain converges to the uniform distribution over
downsets $\Omega_{mon}$.

We now define the {\it biased} Markov chain $\mmon$ by using Metropolis-Hastings transition probabilities so that we converge to the desired distribution on biased surfaces. This new chain connects the state space by the same argument as in Lemma~\ref{conlemma}.

\vspace{.2in}
\noindent  {\bf The Biased Markov chain $\mmon$} 

\vspace{.1in}
\noindent {\tt Starting at any $\sigma_0$, iterate  the following:} 

\begin{itemize}
\item {\tt Choose a diagonal $d$ and a bit $b\in\{\pm 1\}$ as described above.}

\item {\tt If $b=+1$, add the cube above and to the right of the vertex $v$ to create $\sigma_{t+1}$, if this is a valid downset.}

\item {\tt  If $b=-1$, the with prob. $\frac{1}{\lambda_x}$ let $\sigma_{t+1}$ be obtained from $\sigma_t$ by removing the cube $x$ below and to the left of $v$ if this is a valid downset. }

\item {\tt Otherwise let $\sigma_{t+1}=\sigma_t$.}

\end{itemize}

The biased Markov chain $\mmon$ converges to the correct distribution on $\Omega_{mon}$
by the detailed balance condition.
Moreover, notice that for any given diagonal $d$ there is at least one choice of $b$ that proposes a move (adding
or removing) that does not result in a valid downset. Therefore $\p[\sigma_{t+1}=\sigma_t]\geq 1/2$, so $\mmon$ is a lazy chain.

\section{Path coupling with exponential metrics}\label{Path Coupling}
Path coupling is a standard technique used to bound mixing times, and although a naive application of it is not sufficient here, we will see that with some new ideas, we can make it work.  One of the innovations behind our proofs is to introduce a new metric, and in some cases this metric requires a modified Path Coupling theorem.  We present the background and our new theorem here.

A {\it coupling} of a chain $\cM$ is a Markov process on $\Omega \times \Omega$ such that the marginals each agree with~$\cM$ and, once the two coordinates coalesce, they move in unison thereafter. The Coupling Lemma bounds the total variation distance by the probability that the processes have coalesced (see, for example,~\cite{ald}):
\begin{theorem}\label{couplinglemma2} For any coupling, we have
$d_{tv}(P^t(x,\cdot),\pi) \leq P(X_t\neq Y_t).$
\end{theorem}

\begin{definition}
For initial states $x,y$ let $$
   T^{x,y}=\min\{t:X_t=Y_t\mid X_0=x,\,Y_0=y\}, $$
and define the
\emph{coupling time} to be $T=\max_{x,y}\E [T^{x,y}]$.
\end{definition}

\noindent The following lemma bounds the mixing time in terms of the coupling time of any coupling (see, for example, ~\cite{ald}):
\begin{theorem}\label{couplinglemma}
For any coupling with coupling time $T$, the mixing time satisfies $\tau(\epsilon)\leq \lceil T e \ln \epsilon^{-1} \rceil.$
\end{theorem}

The goal, then, is to define a good coupling and show that the coupling time is polynomially bounded.
Path Coupling is a convenient way of establishing this property by only considering a subset 
of the joint state space $\Omega \times \Omega.$

\begin{theorem}{\bf (Dyer and Greenhill \cite{DG97})}\label{PathCoupThm}
Let $\varphi$ be an integer valued metric defined on $\Omega \times\Omega$ which takes values in $\{0,\dots,B\}$. Let $U$ be a subset of $\Omega\times\Omega$ such that for all $(x_t,y_t)\in\Omega\times\Omega$ there exists a path $x_t=z_0,z_1,\dots,z_r=y_t$ between $x_t$ and $y_t$ such that $(z_i,z_{i+1})\in U$ for $0\leq i < r$ and $\sum_{i=0}^{r-1} \varphi(z_i,z_{i+1}) = \varphi(x_t,y_t).$
Let ${\cal M}$ be a Markov chain on $\Omega$ with transition matrix $P$. Consider any random function $f: \Omega \rightarrow \Omega$ such that $\p[f(x)=y]=P(x,y)$ for all $x, y \in \Omega$, and define a coupling of the Markov chain by $(x_t,y_t) \rightarrow (x_{t+1},y_{t+1})=(f(x_t),f(y_t))$.  Suppose there exists $\beta\leq 1$ such that
 $$\E[\varphi(x_{t+1},y_{t+1})] \leq \beta \varphi(x_t,y_t),$$
for all $(x_t,y_t)\in U$.
\begin{enumerate}
\item
If $\beta < 1$, then the mixing time satisfies
$$ \tau(\epsilon) \leq \frac{\ln (B \epsilon^{-1})}{1-\beta}.$$
\item 
If $\beta=1$ (i.e., $\E[\Delta \varphi(x_t,y_t)] \leq 0,$ for all $x_t,
y_t \in U),$
let $\alpha >0$ satisfy
$\hbox{Pr}[\varphi(x_{t+1},y_{t+1}) \neq \varphi(x_t,y_t)]$ 
$\geq \alpha$ for all $t$ such
that $x_t \neq y_t$.
The mixing time of ${\cal M}$ then satisfies
$$\tau(\epsilon)\leq \Bigl\lceil \frac{\e B^2}{\alpha}\Bigr\rceil\lceil\ln
\epsilon^{-1}\rceil.
$$
\end{enumerate}
\end{theorem}

To understand why it is difficult to use coupling to prove that $\mmon$ is rapidly mixing, 
we first examine the straightforward coupling of $(\sigma_t,\rho_t)$ in the uniform bias case.  The natural coupling
simply chooses the same diagonal $d$ and bit $b$ to generate both $\sigma_{t+1}$ and $\rho_{t+1}$. We first consider a natural distance metric on $\Omega_{mon}\times\Omega_{mon}$ called the Hamming distance, where $h(\sigma_t,\rho_t)=|\sigma_t\tri\rho_t|$ (and $\tri$ is the symmetric difference). However, with this coupling and metric, we face difficulty with even the simplest of pairs $(\sigma_t,\rho_t)$.

\condcomment{\boolean{includefigs}}{ 
\begin{figure}[ht]
\centering
\includegraphics[scale=.5]{2d2.pdf}\hspace{.5in}\includegraphics[scale=.5]{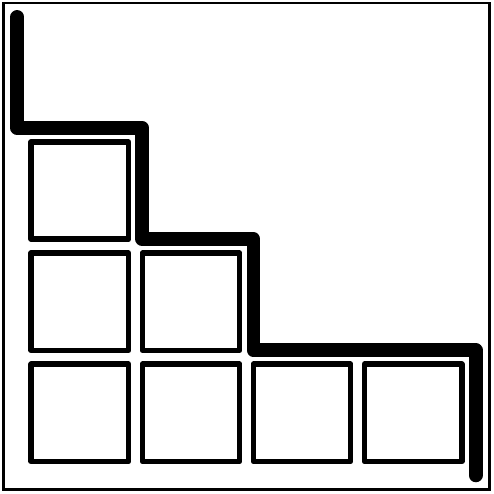}
\caption[Path Coupling of Monotonic Surfaces in Two Dimensions]{A pair of downsets $\sigma_t$ (left) and $\rho_t$ (right) where $\rho_t=\sigma_t\union\{(2,2)\}$ .}
\label{2dCoupProbFig}
\end{figure}
}

Examine the pair of downsets in Figure \ref{2dCoupProbFig}. They differ on a single point, so $h(\sigma_t,\rho_t)=1$. In order to use the coupling theorem above, the expected distance $\E[h(\sigma_{t+1},\rho_{t+1})]$ must be at most $h(\sigma_t,\rho_t)$. For this pair of downsets, there are two moves that decrease that distance; if $\mmon$ chooses the diagonal $d_0=\{(0,0)+t\ubar^*: t\geq 0\}$ and either $b=+1$ or $b=-1$, then $(\sigma_{t+1},\rho_{t+1})$ is $(\rho_t,\rho_t)$ or $(\sigma_t,\sigma_t)$, respectively. In either case, the distance between $\sigma_{t+1}$ and $\rho_{t+1}$ decreases by $1$.
There are also two moves that \emph{increase} the distance. If $\mmon$ chooses $d=\{(1,0)+t\ubar^*:t\geq 0\}$ and $b=+1$ or $d=\{(0,1)+t\ubar^*:t\geq 0\}$ and $b=+1$, then $\rho_{t+1}$ gains a new point ($(3,2)$ or $(2,3)$, respectively), but $\sigma_{t+1}$ remains unchanged; no addition to $\sigma_t$ of a vector along that diagonal leaves a valid downset. With either of these choices, the distance between $\sigma_{t+1}$ and $\rho_{t+1}$ increases by $1$.
If $\lambda=1$, this is sufficient for coupling; the expected change in distance is $0$. Unfortunately, for any $\lambda>1$, the two bad moves happen with probability $1$, whereas the two good moves happen with probability $1$ and $1/\lambda$, respectively.  Therefore the expected distance between the pair $(\sigma_t,\rho_t)$ \emph{increases} after one step.
In higher dimensions, the situation becomes even worse. For the pair of 3 dimensional downsets in Figure \ref{3dCoupProbFig}, there are \emph{three} moves which increase the Hamming distance and %(Choosing $((0,1),+1)$ adds $(1,1,0)$ to $\rho_t$, choosing $((-1,1),+1)$ adds $(1,0,1)$ to $\rho_t$, and choosing $((1,0),+1)$ adds $(0,1,1)$ to $\rho_t$. In all three cases the move succeeds with probability $1$, and $\sigma_t$ remains unchanged). 
only two moves which decrease the distance. %(choosing $((0,0),+1)$ makes $\rho_{t+1}=\sigma_t$ and choosing $((0,0),-1)$ can make $\sigma_{t+1}=\rho_t$).
 Of course, the three moves that increase the distance succeed with probability $1$, but one of the two moves which decreases the distance only succeeds with probability ${1}/{\lambda}$.

 \begin{figure}[ht]
\centering
\includegraphics[scale=.2]{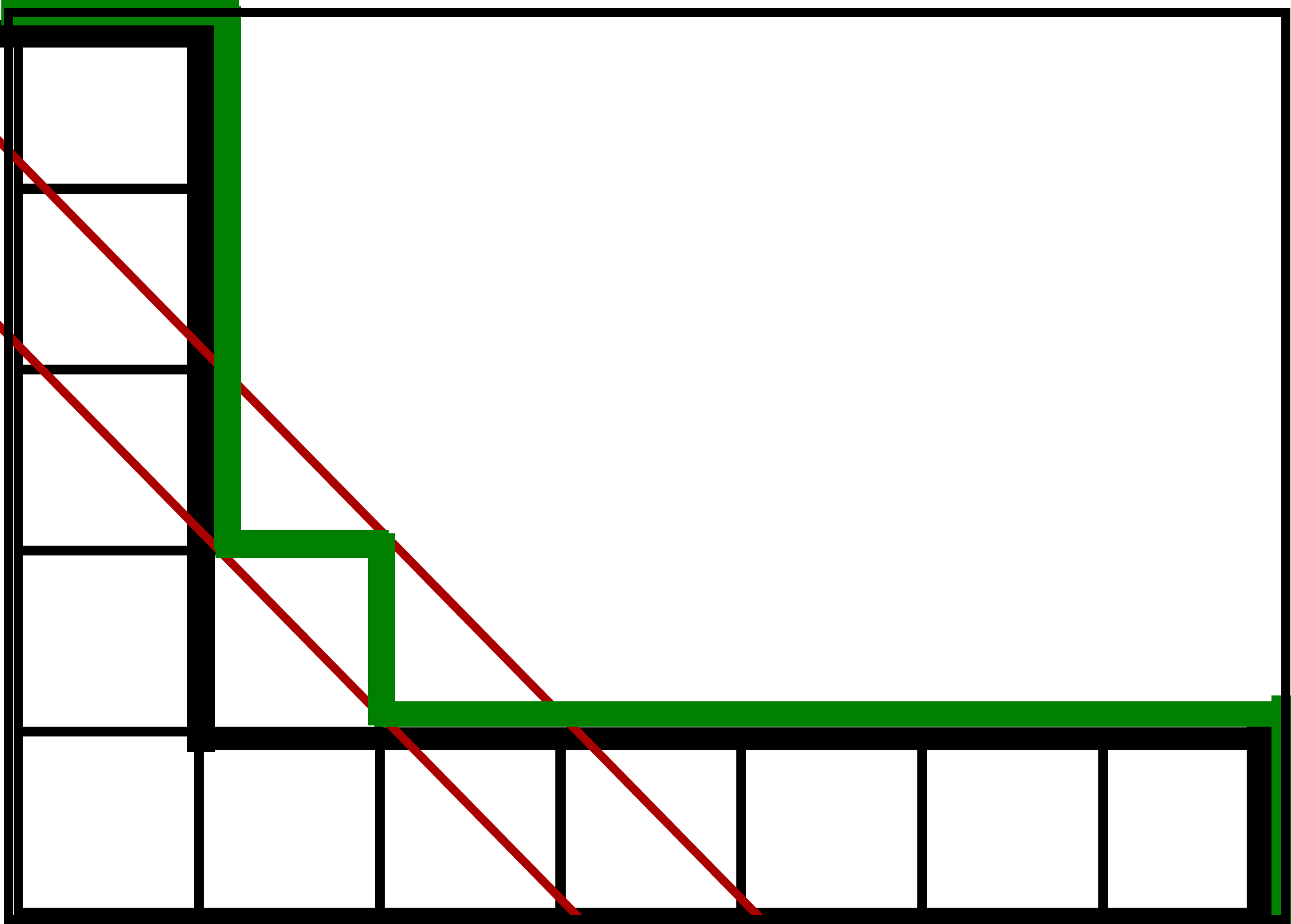}
\put(-130,52){$\mu^{k}$}
\put(-137,64){$\mu^{k-1}$}
\put(-64,-20){(a)}
\hspace{.8in}
\includegraphics[scale=.2]{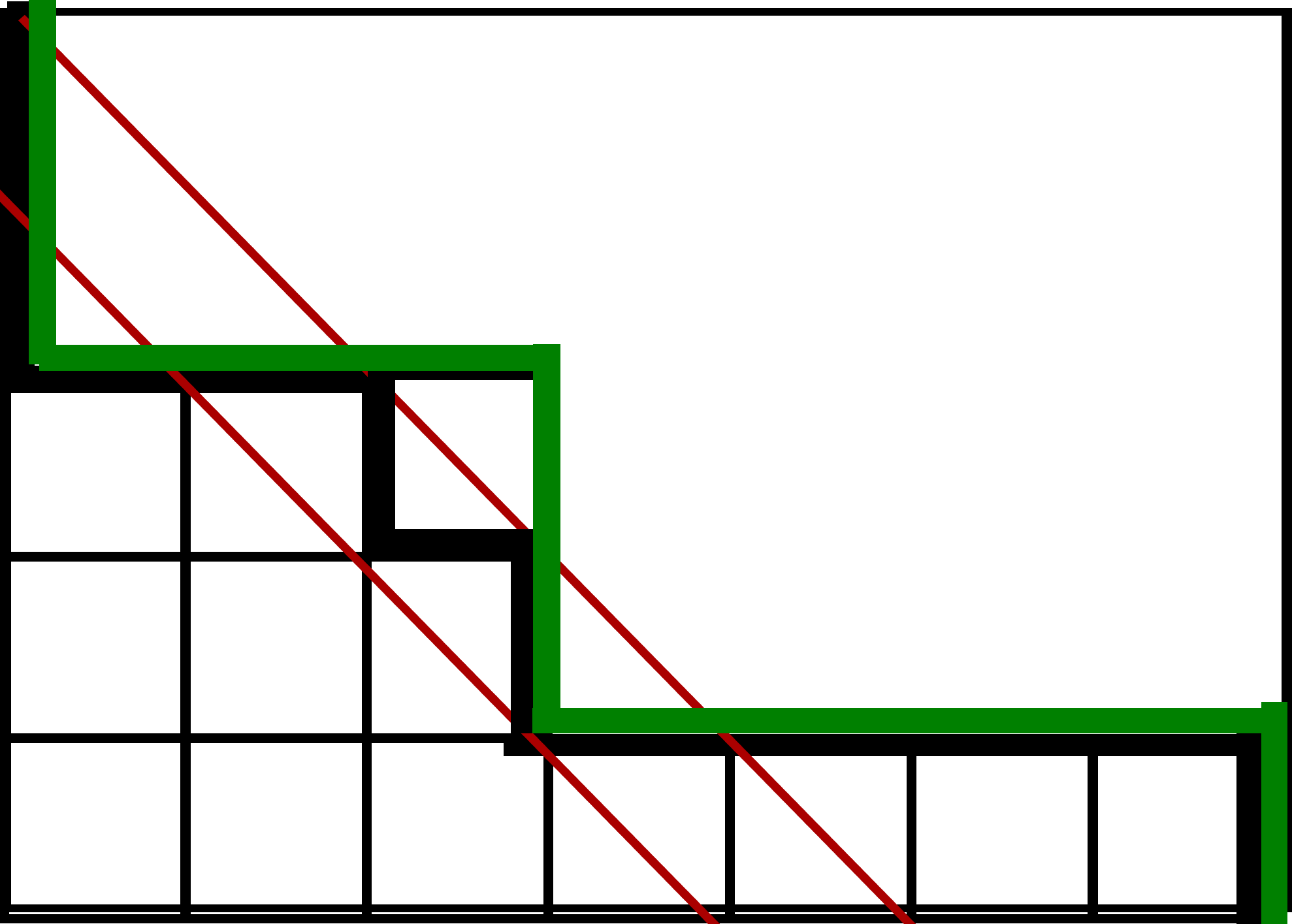}
\put(-137,64){$\mu^{k-1}$}
\put(-137,76){$\mu^{k-2}$}
\put(-64,-20){(b)}
\caption[Path coupling with an exponential distance metric]{ Path coupling with an exponential distance metric.}
\label{distmetric23a}
\end{figure}
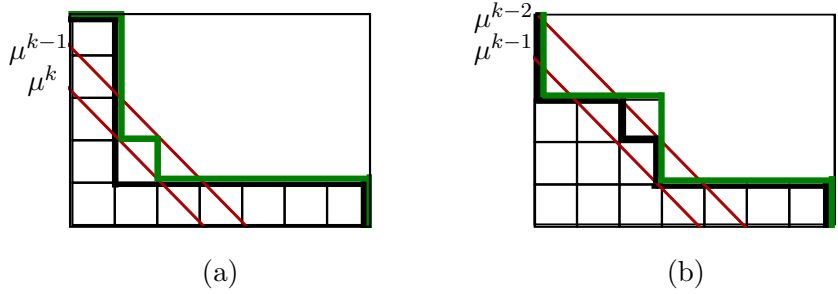

One promising remedy is to alter the distance metric. The bad cases described above involve two downsets that differ on some point $x$, where the two moves which decrease the distance involve removing $x$ from $\sigma_t\tri\rho_t$, while the moves that increase the distance involve adding $x+\ubar_i$ to $\sigma_t\tri\rho_t$ for some $i$ (where $\ubar_1=(1,0)$ and $\ubar_2=(0,1)$).  Since the bad moves happen with greater probability than the good moves, we consider a distance metric that counts the distance between two sets that differ on $x$ as greater than the distance between two sets that differ on $x+\ubar_i$.  Specifically, we give a different weight to each Northwest-Southeast diagonal, with the weights smaller along higher diagonals (as in Figure~\ref{distmetric23a}a).  This allows us to make the change in distance nonpositive in the above cases.  Of course we must ensure that the difference in weight is not too great.  This is because the opposite situation might happen as well, where the two bad moves involve removing $x-\ubar_i$ for some~$i$, (as in Figure~\ref{distmetric23a}b); although this situation was not a problem for the Hamming distance metric, if we assign too much weight to those bad moves, the change in distance might be positive in this case.
We find the following distance metric suffices.  First, let $R$ be the $h\times w$ rectangle in $\Z^2$ and define $\mu=\sqrt{\lambda}\geq 1$.  Then for two downsets $\sigma,\rho$ in $R$, let
\begin{equation}\label{distmetric}
\phi(\sigma,\rho)=\mu^{w+h}\sum_{x\in\sigma\oplus\rho}\mu^{- \|x\|_1,}
\end{equation}
where $\|\cdot\|_1$ is the $L_1$ norm.  Notice that all elements on each Northwest-Southeast diagonal have the same $L_1$ norm, and so this metric assigns a weight of $\mu^k$ for the $k$th diagonal from the top right, as in Figure~\ref{distmetric}.  Notice also that this definition ensures that the distance between any two downsets is either $0$ or at least $1$.  We present the proof that this metric is decreasing in
expectation at every step in Section~\ref{FixedExpProof}.

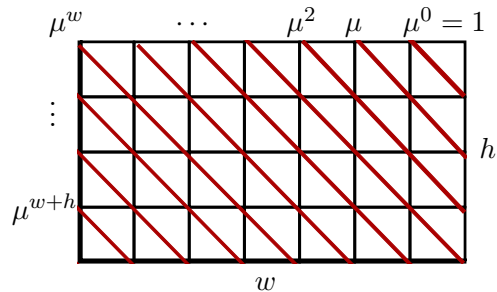
\begin{figure}[ht]
\centering
\includegraphics[scale=.26]{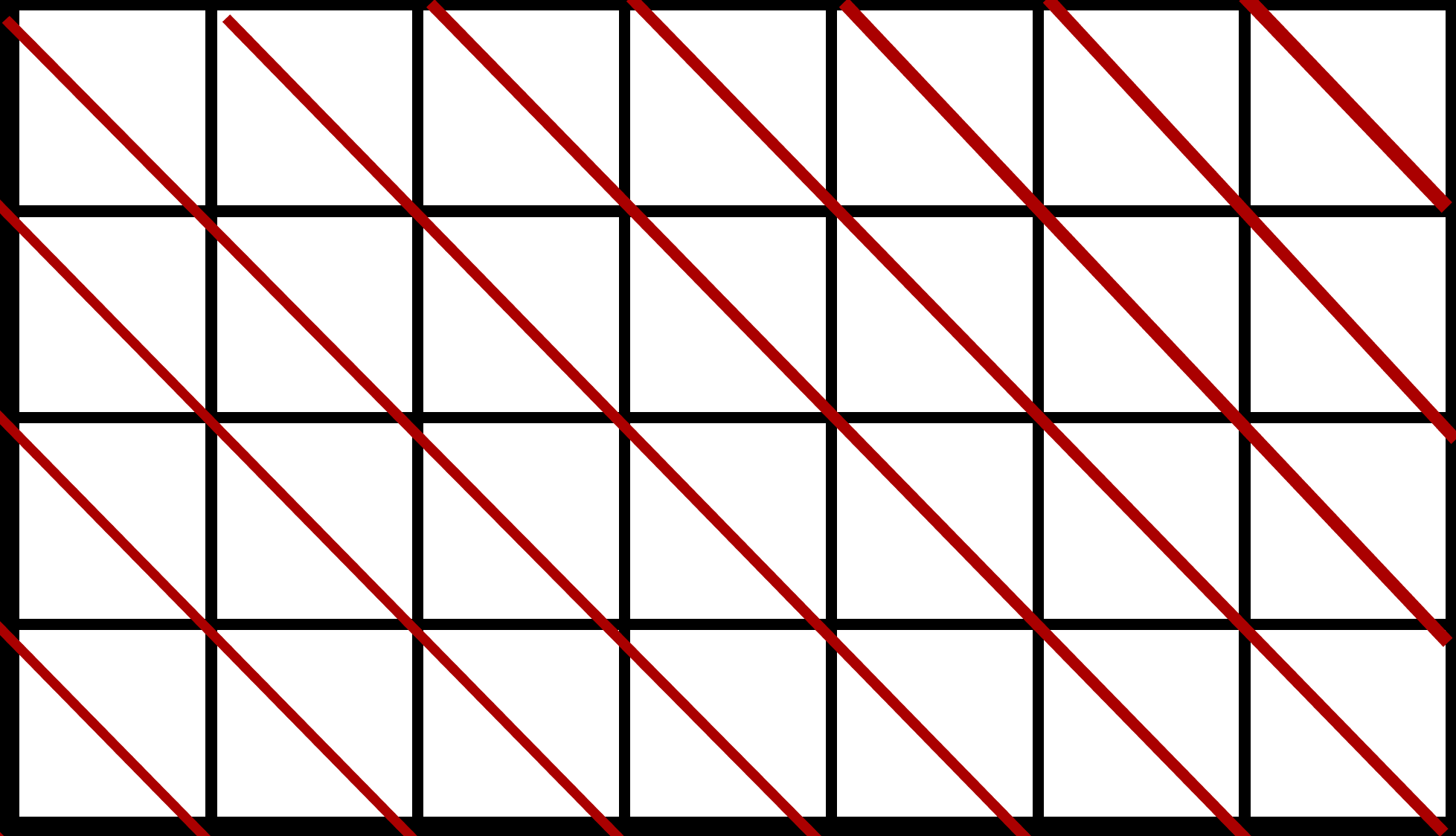}
\put(-24,88){$\mu^0=1$}
\put(-46,88){$\mu$}
\put(-68,88){$\mu^2$}
\put(-110,88){$\cdots$}
\put(-158,88){$\mu^w$}
\put(-158,52){$\vdots$}
\put(-172,17){$\mu^{w+h}$}
\put(-80,-10){$w$}
\put(5,40){$h$}
%\hspace{.5in}
%\includegraphics[scale=.5]{2d.pdf}
\caption[Exponential distance metric]{ Exponential distance metric.}
\label{distmetric}
\end{figure}

Unfortunately, this definition of the distance metric presents new problems. First, the distances might now take on non-integer values, while the Path Coupling Theorem requires integer valued metrics. In fact, if this restriction is merely removed, then the theorem is no longer true as the distances might get smaller and smaller without coalescence occurring in a polynomial number of steps. However, it is enough to add the additional condition that no pairs of configurations can have a distance within the open interval $(0,1)$. 
The second, more serious concern is that the maximum distance between two configurations can be exponentially large in $n$. If the distance only changes by a small (polynomial) amount in 
each step, then we cannot expect the distance to be zero in only a polynomial number of steps.
For example, for small $\lambda$ we can find configurations
$x_t$ and $y_t$ so that 
$\E[\varphi(x_{t+1},y_{t+1})] \leq (1-2^{-n}) \varphi(x_t,y_t),$ so 
the expected change is too small to apply the first part of Theorem~\ref{PathCoupThm}. Moreover,
the maximum distance $B$ is very large, so we cannot get a good bound on the mixing time
using the second part of Theorem~\ref{PathCoupThm} either.

%\subsection{Path coupling with exponential metrics}
The following modification of the Path Coupling Theorem allows us to handle cases when
the distances can be exponentially large and the expected change in distance is small (or even zero).
We show that it suffices to prove that the expected change in the absolute value of the distance
is proportional to the current distance, and with this condition the mixing time is polynomially
bounded. We apply this new theorem to the biased Markov chain $\mmon$ in Section~\ref{FixedExpProof}.

\begin{theorem}
\label{GeomThm}
Let $\phi$ be a metric defined on $\Omega \times \Omega$ which takes finitely many values in $\{0\}\cup[1, B]$. Let $U$ be a subset of $\Omega\times\Omega$ such that for all $(X_t,Y_t)\in \Omega\times\Omega$ there exists a path $X_t=Z_0,Z_1,\ldots,Z_r=Y_t$ such that $(Z_i,Z_{i+1})\in U$ for $0\leq i<r$ and $\sum_{i=0}^{r-1} \phi(Z_i,Z_{i+1})=\phi(X_t,Y_t).$

Let $\cM$ be a lazy Markov chain on $\Omega$ and let $(X_t,Y_t)$ be a coupling of $\cM$, with $\phi_t=\phi(X_t,Y_t)$. Suppose there exists $\beta\leq 1$ such that, for all $(X_t,Y_t)\in U$,
$$\E[\phi_{t+1}]\leq \beta \phi_t.$$
\begin{enumerate}

\item
If $\beta<1$, then the mixing time satisfies
$$\tau(\varepsilon)\leq \frac{\ln(B\varepsilon^{-1})}{1-\beta}.$$

\item
If there exists $\kappa,\eta\in(0,1)$ such that $\p\left[|\phi_{t+1}-\phi_t|\geq \eta \phi_t\right]\geq \kappa$
for all $t$ provided that $X_t\neq Y_t$, then
$$\tau(\varepsilon)=O\left(\frac{ \ln^2 B\ln\varepsilon^{-1}}{ \ln^2(1+\eta) \kappa}\right).$$
%$$\tau(\varepsilon)\leq\left\lceil\frac{3e \ln^2(B)}{ \ln^2(1+\eta) \kappa}\right\rceil\left\lceil \ln\varepsilon^{-1}\right\rceil.$$
\end{enumerate}
\end{theorem}

There are two important differences between Theorem~\ref{PathCoupThm} and 
Theorem \ref{GeomThm}. The first is that Theorem \ref{GeomThm} allows for non-integer metrics (provided that for all $X,Y\in\Omega$, $\phi(X,Y)<1$ implies $\phi(X,Y)=0$). This is a minor restructuring of the proof of Theorem \ref{PathCoupThm} \cite{lrs}, and follows exactly from their method.
The second is that $\beta$ may equal $1$ while $B$ is exponentially large; this is the case in which both parts of Theorem \ref{PathCoupThm} were unable to prove rapid mixing. This can be shown again with a slight modification of the original proof, essentially replacing the original distance $\phi(X_t,Y_t)$ with $\ln(\phi(X_t,Y_t))$. There are some technical details concerning the expectation and variance of the logarithm, but the novelty of Theorem \ref{GeomThm} is more in the statement of the result than a new method of proof.

Note that including this case of $\beta=1$ and exponential $B$ requires a strong bound on the variance of $\phi_t$. Without this bound on variance, Theorem \ref{GeomThm} is not true; if $\phi_0=2^n$ and $\phi_{t+1}=\phi_t-1$ for all $t\geq 1$, then clearly it will take time exponential in $n$ for $\phi_t=0$.

In order to prove Theorem \ref{GeomThm} to handle exponential metrics, we define a new variable $\psi$, which is essentially $\ln(\phi)$. However, if we hope to prove rapid mixing by looking at $\ln(\phi)$, we need to bound the time to reach $\ln(0)=-\infty$, and the expected time could be unbounded. In particular, in order to prove rapid mixing, we need that the sequence $\{\psi_t\}$ has bounded differences. The technical fix that we make relies on the assumption that $\phi_t\notin (0,1)$, so we need only bound the time until we reach a negative value for $\ln(\phi_t)$. Hence we define
$$\psi_t = \begin{cases} \ln(\phi_t) & \text{ if }\phi_t > 0 \\ -2\ln 2 & \text{ if }\phi_t=0\end{cases}.$$
This means that $\psi_t\in[-2\ln 2,\ln B]$. The particular value at zero is chosen so that if the expected distance $\phi_t$ is non-decreasing, then the expected value of $\psi_t$ is non-decreasing, and that if the variance of $\phi_t$ is at least a linear factor, then the variance of $\psi_t$ is at least a constant.

The following Martingale Lemma follows the proof of Lemma 6 in \cite{lrs}.

\begin{lemma}
\label{MartingaleLemma}
Given any bounded function $\phi(t),$ with $d\leq \phi(t)\leq D$ for some $d,D\in \R$ and for all $t\geq 0$, and a stopping value $q$, let $T=\min\{t:\phi(t)=q\}$.  If, for all $t\geq 0$, we have $\E[\phi(t+1)-\phi(t)]\leq 0$ and $\E[(\phi(t+1)-\phi(t))^2]\geq Q$ for some $Q$, then
$$\E[T]\leq \frac{2 D^2+ q^2 -2qD}{Q}.$$
\end{lemma}
\begin{proof}
Define the process $Z(t):=(D - \phi(t))^2 - Qt$. Examining the expected difference between $Z(t)$ and $Z(t+1)$, we have
\begin{align*}
\E[Z(t+1)-Z(t)]&= \E[(D- \phi(t+1))^2 - (D -\phi(t))^2] - Q\\
&= \E[-2D(\phi(t+1)-\phi(t)) + \phi(t+1)^2 - \phi(t)^2] - Q\\
&= 2(\phi(t)-D) \E[\phi(t+1)-\phi(t)] +  \E[(\phi(t+1)-\phi(t))^2]  - Q \geq 0.
\end{align*}
Also, since the differences $Z(t+1)-Z(t)$ are bounded, so $\{Z(t)\}$ is a submartingale.  $T$ is a stopping time for $Z(t)$, so we may apply the Optional Stopping Theorem for submartingales to deduce that
\begin{align*}
\E[T]&\leq \frac{1}{Q}\left[ \phi(0)(2D-\phi(0)) +q^2-2qD\right]\\
&\leq \frac{2D^2 + q^2 -2qD}{Q}.
\end{align*}
\end{proof}

Now we may prove the exponential metric theorem, Theorem~\ref{GeomThm}.
\begin{proof}
\noindent {\bf Part (1)} \ \ This case follows directly from the proof of Theorem~\ref{PathCoupThm}, while allowing non-integer valued metrics. Since $\E[\phi_{t+1}]\leq \beta \phi_t$ for all $t$, it follows that
\begin{eqnarray}\label{phibound}
\E[\phi(X_t,Y_t)]=\E[\phi_{t}]\leq \beta^t \phi_0=\beta^t \phi(X_0,Y_0)\leq \beta^t B.
\end{eqnarray}
Since $\phi_t$ is nonnegative, takes values in $\{0\}\cup [1,B]$ and is equal to zero whenever $X_t=Y_t$, we have
\begin{eqnarray*}
\E[\phi_t] &=& \int_1^B x\text{Pr}(\phi_t=x) dx\\
&\geq & \int_1^B \text{Pr}(\phi_t=x) dx\\
&= & \text{Pr}(\phi_t\geq 1)\\
&= & \text{Pr}(X_t\neq Y_t).
\end{eqnarray*}
Then since by Equation~\ref{phibound}, $\E[\phi_t]\leq \epsilon$ whenever $t\geq \log(B\epsilon^{-1})/\log(\beta^{-1})$,
 the Coupling Lemma~\ref{couplinglemma2} implies
$$\tau(\epsilon)\leq \ln(B\epsilon^{-1})/\ln(\beta^{-1}).$$
Since $\ln(\beta^{-1})>1-\beta$, Part (1) follows.

\vspace{.2in}

\noindent {\bf Part (2)}\ \ For part (2), we will show that $\psi_t$ satisfies the conditions of Lemma~\ref{MartingaleLemma}, with $q=-2\ln 2$, $D=\ln B$ and $Q=\ln^2(1+\eta)\kappa$.  Note that this proves the theorem, since Theorem~\ref{couplinglemma} implies that 
$$\tau(\epsilon) \leq  \lceil \E[T] e \ln \epsilon^{-1} \rceil \leq \left\lceil \frac{2 D^2+ q^2 -2qD}{Q}  e \ln \epsilon^{-1} \right\rceil = O\left( \frac{\ln^2 B\ln\epsilon^{-1}}{\ln^2(1+\eta)\kappa}\right).$$
%\left \lceil \frac{{3e(\ln B)^2}}{\ln(1+\eta)^2\kappa}\right\rceil \lceil \ln(\epsilon^{-1})\rceil.$$

First we show that since $\E[\phi_{t+1}-\phi_t]\leq 0,$ then also $\E[\psi_{t+1}-\psi_t] \leq 0.$  We may assume $\phi_t\neq 0$.
Given the value of $\phi_t$, let $\{r_0,r_1,r_2, \dots, r_N\}$ be the possible values for $\phi_{t+1}$, each occurring with probability $\{\zeta_0, \zeta_1,\zeta_2, \dots, \zeta_N \}$. That is, $\p[\phi_{t+1}=r_i | \phi_t]=\zeta_i$, with $\sum_{i=0}^N \zeta_i = 1$. Assume $r_0=0$.
As our chain is lazy, $\p[\phi_{t+1}=\phi_t]\geq 1/2$. Therefore $\zeta_0\leq 1/2$. Now,
\begin{align*}
\E[\psi_{t+1}|\psi_t] &= \zeta_0 (-2\ln2) + \sum_{i=1}^N \zeta_i \ln_2(r_i)\\
               &=-2\ln2\zeta_0+ \ln\left(\prod_{i=1}^N r_i^{\zeta_i}\right)\\
               &\leq -2\ln2\zeta_0 + \ln\left(\frac{\sum_{i=1}^N \zeta_i r_i}{1-\zeta_0}\right)\\
               &= \ln\left(\E[\phi_{t+1}|\phi_t]\right) -2\ln2\zeta_0 - \ln(1-\zeta_0)\\
               &\leq \ln\left(\E[\phi_{t+1}|\phi_t]\right)\\
               &\leq \ln \phi_t \ = \ \psi_t,
\end{align*}
where the first inequality is by the Arithmetic-Geometric Mean Inequality, and the second follows from the fact that ${\ln(1-\zeta_0)}/{\zeta_0}\geq -2\ln 2$ for $\zeta_0\in (0,\frac{1}{2})$.

Next we prove that if there exist constants $\kappa, \eta\in(0,1)$ such that $\p[|\phi_{t+1}-\phi_t|\geq \eta\phi_t]\geq \kappa$ for $\phi_t\not=0$, then
$$\E[|\psi_{t+1}-\psi_t|]\geq \ln(1+\eta)\kappa + \ln 2\p[\phi_{t+1}=0].$$
Let $\zeta_0=\p[\phi_{t+1}=0]$.  Then
\begin{align*}
\kappa &\leq \p[|\phi_{t+1}-\phi_t|\geq \eta \phi_t]\\
       &=1\cdot \zeta_0 + \p[|\phi_{t+1}-\phi_t|\geq \eta \phi_t|\phi_{t+1}\neq 0] (1-\zeta_0).
\end{align*}
Now because, by definition, $\psi_{t+1}=\ln(\phi_{t+1})$ when $\phi_{t+1}\neq 0$, we have
\begin{align*}
\p[|\phi_{t+1}&-\phi_t|\geq \eta \phi_t|\phi_{t+1}\neq 0]\\
&= \p\left[\frac{\phi_{t+1}}{\phi_t}-1 \geq \eta|\phi_{t+1}\neq 0\right]+ \p\left[\frac{\phi_{t+1}}{\phi_t}-1 \leq- \eta|\phi_{t+1}\neq 0\right]\\
&= \p\left[\psi_{t+1}-\psi_t \geq \ln(1+\eta )|\phi_{t+1}\neq 0\right]+\p\left[\psi_{t+1}-\psi_t \leq \ln(1-\eta)|\phi_{t+1}\neq 0\right]\\
&\leq \p\left[|\psi_{t+1}-\psi_t |\geq \ln(1+\eta)|\phi_{t+1}\neq 0\right].
\end{align*}
Since $\phi_t\geq 1$, we have $\psi_t\geq 0$, so $|-2\ln2-\psi_t|\geq 2\ln 2$. Note that $m:=\ln^2(1+\eta)< \ln^2(2)$, since $\eta<1$. This yields
\begin{align*}
\E[(\psi_{t+1}-\psi_t)^2]&=(-2\ln 2-\psi_t)^2\zeta_0+\sum_{\ell\in\Omega, \ell\neq 0} \left(\ln(\ell)-\ln \phi_t\right)^2 \p[\phi_{t+1}=\ell]\\
&\geq (2\ln 2)^2 \zeta_0+ m \p\left[|\psi_{t+1}-\psi_t |\geq m|\phi_{t+1}\neq 0\right] (1-\zeta_0)\\
&\geq (2\ln 2)^2 \zeta_0 + m \left( \frac{ \kappa -\zeta_0}{1-\zeta_0}\right) (1-\zeta_0)\\
&= \left((2\ln 2)^2-m\right)\zeta_0 + m\kappa\\
&> m \kappa + 3 \zeta_0 \ln^2 2.
\end{align*}
Hence we have $Q=\ln^2(1+\eta)\kappa < \E[(\psi_{t+1}-\psi_t)^2],$
as desired.  %, so we may apply Lemma~\ref{MartingaleLemma} to get
%$$E[T]\leq \frac{2\ln^2 B}{\ln^2(1+\eta)\kappa}.$$

\end{proof}

\section{Fast mixing of the uniform bias Markov chain}\label{UnifSection}
We start by looking at the biased Markov chain $\mmon$ when the biases are uniform.  We use the exponential metric introduced in Section~\ref{Path Coupling} to show that $\mmon$ is rapidly mixing whenever $\lambda\geq d^2$ for arbitrary dimension $d$, and for all $\lambda>1$ when $d=2$.
In Section~\ref{ceiling}, we  present a simple hitting time argument that proves the biased chain converges in polynomial time, as long as the minimum bias is at least $d$ and the region is a $d-$dimensional hypercube.
We conjecture that the chain is rapidly mixing for all values of a uniform bias $\lambda > 1$ in all dimensions $d \geq 2$, but do not yet have a proof for small values of $\lambda$ in dimensions higher than 2.

\subsection{Exponential metric for the uniform bias chain}\label{FixedExpProof}
First, we use our exponential distance metric to analyze $\mmon$.
We show that in two dimensions, the biased chain is rapidly mixing for any uniform bias, even when the bias is arbitrarily close to one.  Our proof is significantly simpler than the arguments of Benjamini et al., while achieving the same optimal bounds on the mixing time for square regions when the bias is constant. In fact, on rectangular $h\times w$ regions of $\Z^2$, where $h\leq w$, we get improved bounds of $O(w(h+\ln w))$, which is optimal when $h=\Omega(\ln w)$.  Specifically, we prove the following theorem for rectangular regions:
\begin{theorem}\label{2DSimple}
Let $R$ be a rectangular $h\times w$ region in $\Z^2$ with uniform bias $\lambda\geq 1$.   Suppose without loss of generality that $h\leq w$.  Let $\chi=\sqrt{\lambda}-1$.  Then 
\begin{enumerate}
\item If $\chi> 0$, then the mixing time of $\mmon$ satisfies $\tau(\varepsilon)=O\left(\chi^{-2}w (h + \ln w) \ln \varepsilon^{-1}\right).$
\item If $\chi\geq 0$, then $\tau(\varepsilon)=O \left( w^3(h + \ln w)^2\ln\varepsilon^{-1}\right).$
\end{enumerate}
\end{theorem}
\noindent If $\lambda>1$ is a constant and $h=\Theta(w)$, then part 1 of this theorem applies and gives a bound of $O(n)$, where $n=hw$ is the area of the region.  On the other hand, for $\lambda$ very close to 1, part 2 provides a polynomial bound on the mixing time.  We prove similar bounds for all \emph{nice} regions, to be defined in Section~\ref{generalization}; essentially, these regions are simply-connected and have no holes.

In higher dimensions, we show the chain is rapidly mixing on $d$-dimensional lattice regions provided the bias $\lambda \geq d^2$.  Again, our bounds on the mixing time are optimal when the regions are hyper-cubes, and we show the chain is rapidly mixing for a large family of  simply-connected regions.  For hypercubes we show:
\begin{theorem}\label{HigherDSimple}
Let $R$ be the d-dimensional $h\times h\times\cdots\times h$ hypercube with volume $h^d=n$ and uniform bias $\lambda\geq d^2$.   Then 
the mixing time of $\mmon$ satisfies $\tau(\varepsilon)=O\left(n\ln \varepsilon^{-1}\right).$
\end{theorem}

For clarity of explanation, we will first handle the simple case when $R$ is a rectangular region in $\Z^2$.  These ideas generalize easily to more complex regions and higher dimensions, but require some extra terminology, which we define in Section~\ref{generalization}.

\subsubsection{Rectangular regions in two dimensions}

We now show that in two dimensions the distance metric given in Equation~\ref{distmetric} suffices to bound the mixing time of $\mmon$, proving Theorem~\ref{2DSimple}.  The argument uses a coupling of $(\sigma_t,\rho_t)$ that simply supplies the same diagonal $d$ and bit $b$ to both $\sigma_t$ and $\rho_t$. We let $U$ be the set of pairs of downsets that differ on a single cube. However, instead of the Hamming distance, we use the distance metric given in equation~\ref{distmetric}:
$$\phi(\sigma,\rho)=\mu^{w+h}\sum_{x\in\sigma\oplus\rho}\mu^{- \|x\|_1,}$$
where $\mu=\sqrt{\lambda}\geq 1$.
We will show that this distance metric satisfies non-negative contraction in $\phi_t$, which is one of the requirements for Theorem \ref{GeomThm}. However, before we can prove that the distances decrease on average, we examine the moves which can \emph{increase} the distance.
 
 \begin{figure}[ht]
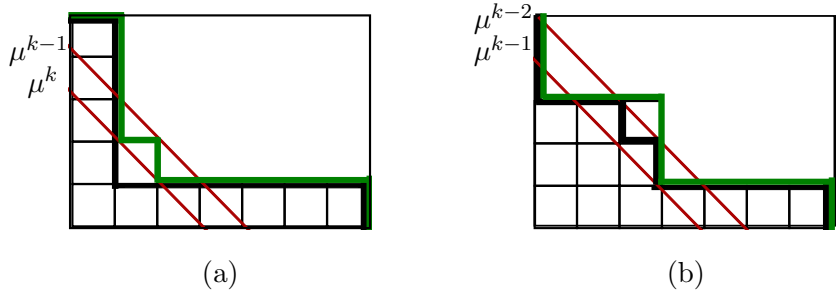

\centering
\includegraphics[scale=.2]{distmetric2.pdf}
\put(-130,52){$\mu^{k}$}
\put(-137,64){$\mu^{k-1}$}
\put(-64,-20){(a)}
\hspace{.8in}
\includegraphics[scale=.2]{distmetric3.pdf}
\put(-137,64){$\mu^{k-1}$}
\put(-137,76){$\mu^{k-2}$}
\put(-64,-20){(b)}
\caption[Path coupling with an exponential distance metric]{ Path coupling with an exponential distance metric.  Two cases where $\sigma_t=\rho_t\union\{x\}$.}
\label{distmetric23}
\end{figure}

For a pair $(\sigma_t,\rho_t)\in U$, there are two different ways the distance can increase in $(\sigma_{t+1},\rho_{t+1})$. If $\sigma_t=\rho_t\union\{x\}$, we can increase the distance by attempting to add a cube $v$ that succeeds in $\sigma_t$ but fails in $\rho_t$, as in Figure~\ref{distmetric23}a. This occurs when $v=x+\ubar_i$ for some $i$, so $v$ is ``supported" in $\sigma_t$ but not $\rho_t$.  Notice that the distance metric $\phi$ gives these bad moves weight that is $\mu$ times smaller than the weight of the two good moves (adding or removing $x$), counteracting their higher probability. The other way to increase the distance between $\sigma_t$ and $\rho_t$ is to remove a $v$ that succeeds in $\rho_t$ but not in $\sigma_t$. This occurs when $v=x-\ubar_i$ for some $i$, as the move creates a valid downset in $\rho_t$ but not in $\sigma_t$, as in Figure~\ref{distmetric23}b. In this case, the distance metric $\phi$ gives these bad moves weight that is $\mu$ times larger than the weight of the two good moves, but for small enough $\mu$, they will still not outweigh the good moves, since the bad moves are less likely to occur than the good moves in this case.
We may now prove Theorem \ref{2DSimple}.

\begin{proof}
We will show that the distance metric $\phi$ defined above satisfies the conditions of Path Coupling Theorem for Exponential metrics, Theorem \ref{GeomThm}.  First we want to show that the expected change in $\phi$ is negative.  There are at most $2$ choices of $(d,b)$ that can increase $\phi_t$.  We claim that each of these has an expected increase of at most $\phi_t\lambda^{-1/2}$.  To see this, consider a move of the form $v=x+\ubar_i$ for some $i$.  Then the increase in distance is $(\sqrt{\lambda})^{w+h - \|v\|_1}={\phi_t}{\lambda}^{-1/2}$.  If the move is of the form $v=x-\ubar_i$ for some $i$, then the increase in distance is ${\lambda}^{-\|v\|_1/2}=\phi_t \sqrt{\lambda}$, but the chance of choosing an appropriate $p$ is ${1}/{\lambda}$. Therefore the expected increase is at most ${\phi_t}{\lambda}^{-1/2}$.

There are also two choices of $(d,b)$ that \emph{decrease} $\phi_t$; corresponding to adding $x$ and removing $x$. These each decrease $\phi_t$ by $\phi_t$, and succeed with probability $1$ and ${1}/{\lambda}$, respectively. Let $\alpha=h+w$; this is the number of choices of diagonals $d$.  The expected change in distance satisfies
\begin{align*}
\E[\phi_{t+1}-\phi_t] &\leq \frac{1}{2\alpha} \left(d\cdot\frac{\phi_t}{\sqrt{\lambda}}-\left(1+\frac{1}{\lambda} \right)\phi_t\right)\\
	&= \frac{\phi_t}{2\alpha} \left(\frac{d}{\sqrt{\lambda}}-1-\frac{1}{\lambda} \right)\\
	&\leq - \frac{\phi_t\chi^2}{2\alpha}.
%	&= -\frac{\phi_t}{2\alpha} \left(\frac{1}{\sqrt{\lambda}} - \frac{d+\sqrt{d^2-4}}{2}\right) \left(\frac{1}{\sqrt{\lambda}} - \frac{d-\sqrt{d^2-4}}{2}\right)\\
\end{align*}

Next,, we check the other conditions of Theorem~\ref{GeomThm}.
For arbitrary $\sigma,\rho\in\Omega_{mon}$, if $x\in\sigma\oplus\rho$ for some $x$, then $\phi(\sigma,\rho)\geq \sqrt{\lambda}^{h+w-\|x\|_1}\geq 1$. Therefore if $\phi(\sigma,\rho)<1$, $\phi(\sigma,\rho)=0$.  Let $U$ be the set of pairs of downsets that differ on a single vector. For arbitrary $\sigma,\rho\in\Omega_{mon}$, we can connect $\sigma$ to $\rho$ by simply adding or removing the vectors in $\sigma\oplus\rho$ one by one, and $\phi(\sigma,\rho)$ is the sum of the distances.  Since the volume of ${R}$ is $n$, there are at most $n$ possible cubes in $\sigma\oplus\rho$, so $\phi(\sigma,\rho)\leq n\lambda^{h/2}$ for all~$\sigma,\rho$.  

We consider two cases.  If $\chi>0$, then $\E[\phi_{t+1}]\leq \beta \phi_t$, where $\beta=1-{\chi^2}/{\alpha}$.  Thus, by Theorem \ref{GeomThm}, we have $\tau(\epsilon)=O(\chi^{-2}w (h + \ln w)\ln \epsilon^{-1})$.  
On the other hand, if $\lambda - 1> 0,$ but less than any constant, then we use the second part of Theorem ~\ref{GeomThm}.  For any pair of $\sigma,\rho$, $\mmon$ can always add a vector $v^*$ in their difference that maximizes $\|v\|_1$.  This would change $\phi_t$ by at least $\phi_t/{\alpha}$.  The appropriate $v^*$ is chosen with probability at least ${1}/{\alpha}$ and the appropriate $b$ is chosen with probability ${1}/{2}$ (and every $p$ succeeds when adding).  Therefore there is a ${1}/{(2\alpha)}$ chance of changing~$\phi_t$ by ${\phi_t}/{(2\alpha)}$, i.e. $P\left(|\phi_{t+1}-\phi_t|\geq {\phi_t}/{(2\alpha)}\right)\geq {1}/{(2\alpha)}.$
Hence in this case
$$\tau(\epsilon)=O\left( \frac{\ln^2(B)}{\ln^2(1+\frac{1}{\alpha})\frac{1}{2\alpha}}  \ln \epsilon^{-1}\right)=O( \alpha^3(h +\ln w)^2 \ln \epsilon^{-1}) = O( w^3(h +\ln w)^2 \ln \epsilon^{-1}).$$

\end{proof}

%Generalization to other regions
\subsection{More complex regions and higher dimensions}\label{generalization}
The results of Section~\ref{FixedExpProof} extend to higher dimensions and more general regions.  We need to be a bit more formal, so we begin with several definitions.

Given a simply-connected region $\widehat{R}$ in $\R^d$ that is a union of unit cubes on the integer lattice, we associate to it a point set $R=\widehat{R}\cap \Z^d$.  In this case we say that $R$ is simply-connected.  
In three dimensions, a monotonic surface in $R$ is the union of two-dimensional faces such that any cross-section along an axis-aligned plane is a two-dimensional monotonic surface. Such a surface is illustrated in Figure \ref{MonSurf}b when ${R}$ is a $2\times 2\times 2$ region.  In general, given a hypercubic region ${R} \subset \Z^d$ composed of unit hyper-cubes,
a $d$-dimensional monotonic surface is a set of $(d-1)$-dimensional faces, such that any cross-section along an axis-aligned $(d-1)$-dimensional hyper-plane is a $(d-1)$-dimensional monotonic surface.

We will restrict our focus to a family of simply-connected regions that have favorable properties for the purposes of sampling monotonic surfaces.  In order to define this family, we need a few preliminaries.

\begin{definition}
Let $\ubar^* = (1,1,...,1) \in \Z^d$. For $v\in \Z^d$, we define $\widehat{r}(v) =\{v+k\ubar^*: k\in \R\},$
and we define the {\bemph ray} $r(v)$ to be the set
$$r(v)=\widehat{r}(v)\bigcap \Z^d.$$
\end{definition}

\begin{definition}
A $d$-dimensional simply-connected region $\widehat{R} \subset \R^d$ is {\bemph nice} if, for all $v \in \hat{R}$, $\widehat{R}\bigcap \hat{r}(v)$ is connected.  We call its associated point set $R$ a \emph{nice} region.
\end{definition}
\noindent A nice region $\widehat{R}$ has no holes, and in particular all monotonic surfaces in $\widehat{R}$ are the upper boundary of a subset of cubes in $\widehat{R}$. Note that all hyper-rectangular regions are nice.

Let $\ubar_i$ be the unit vector in the $i$th direction.  Given a nice region $R$, we let $R_L = \{v \in R {\hbox{\ such that \ }} v-\ubar^* \notin R\}$ be the lower envelope of the region.

\begin{definition}
Let $R \subset \Z^d$ be a nice region.
A {\bemph downset} is a subset $\sigma \subseteq R,$ with $R_L \subseteq \sigma,$
such that for any $i$, if $v\in\sigma$ and $v-\ubar_i\in R$, then $v-\ubar_i\in\sigma$.
\end{definition}

\noindent For a nice region $R$, we define the state space $\Omega_{mon}$ to be the set of all downsets of $R$.
We will represent a downset by its \emph{upper boundary}:

\begin{definition} Let $R$ be any nice region and let $\sigma$ be any downset of $R$. We say
the {\bemph upper boundary} of $\sigma$ is $\partial(\sigma) = \{v \in \sigma {\hbox{\ such that \ }} v+\ubar^*
\notin \sigma \}$.
\end{definition}
\noindent The upper boundary is a monotonic surface in bijection with
the downset that defines it.  It is important to notice that for any downset $\sigma$ and point $v \notin \sigma$, if $\sigma \cup v$ is a valid downset, then $|\partial(\sigma)| = |\partial(\sigma \cup v)|$. This is because 
$\partial(\sigma \cup \{v\}) = \partial(\sigma) \cup \{v\} \setminus \{v-\ubar^*\}.$
It follows that for any nice region $R$, the size of the boundary of a valid downset is fixed.
This observation motivates the following two definitions that will be convenient when we state the mixing time of our Markov chain.

\begin{definition}
The {\bemph span} of a nice region $R$ is $\alpha = |\partial(\sigma)|$, for
any downset $\sigma$ of $R$.
\end{definition}

\begin{definition}
Let $R$ be any nice region. The {\bemph stretch} of $R$ is $\gamma=\max_{v\in R} |R\cap r(v)|.$
\end{definition}
\noindent Thus, the stretch is the maximal distance between two points in $R$ in the $\ubar^*$ direction.
Suppose, for example, that $R$ is an $h  \times \dots \times h$ region in $\Z^d$.  Then the
span is $\alpha = d h^{d-1}$ and the stretch is~$\gamma=h$.

The Markov chain operates as follows.  
Starting at an arbitrary downset, e.g., let $\sigma_0=R_L$, where $R_L$ is the empty downset, and repeat
the following steps. If we are at a downset $\sigma_t$ at time $t$, pick a point $v \in \partial(\sigma)$ and an integer $b \in \pm 1$ uniformly at random. If $b=+1$, let $\sigma_{t+1} = \sigma_t \cup (v+\ubar^*)$ if this is a valid downset. If $b=-1$, let $\sigma_{t+1} = \sigma_t \setminus \{v\}$ if this is a valid downset, however in the biased case we will do this move with the appropriate Metroplis-Hastings probabilities.   In all other cases, keep $\sigma_t$ unchanged so that $\sigma_{t+1} = \sigma_t$. 

\vspace{.2in}
\noindent  {\bf The Biased Markov chain $\mmon$} 

\vspace{.1in}
\noindent {\tt Starting at any $\sigma_0$, iterate  the following:} 

\begin{itemize}
\item {\tt Choose $(v,b,p)$ uniformly at random from $\partial(\sigma_t) \times \{+1, -1\} \times (0,1)$.}

\item {\tt If $b=+1$, let $\sigma_{t+1}=\sigma \cup \{v+\ubar^*\}$ if it is a valid downset.}

\item {\tt If $b=-1$ and $p \leq {\lambda_v}^{-1}$, let $\sigma_{t+1} = \sigma \setminus \{v\}$ if it is
a valid downset.}

\item {\tt Otherwise let $\sigma_{t+1}=\sigma_t$.}

\end{itemize}

Observe that the chain connects the state space.  To see this, let $\sigma$ be any downset and let $v^{max}$ be any point
in $\sigma$ such that $\sum_i v_i^{max}$ is maximized.  We can always remove $v^{max}$
and move to $\sigma' = \sigma \setminus v^{max}$ without violating the downset condition. Thus,
from any valid downset $\sigma$ we can always remove points and get to the ``lowest'' downset~$R_L$. Also, such a sequence of steps can be reversed to move from $R_L$ to any other downset $\rho$.
Moreover, the biased Markov chain $\mmon$ converges to the correct distribution on $\Omega_{mon}$
by the detailed balance condition.  Further, we again can see that the chain must be lazy since for each choice of $v$ one of the choices
of $b$ will keep the configuration unchanged.

\condcomment{\boolean{includefigs}}{ 
\begin{figure}[ht]
\centering
\includegraphics[scale=.5]{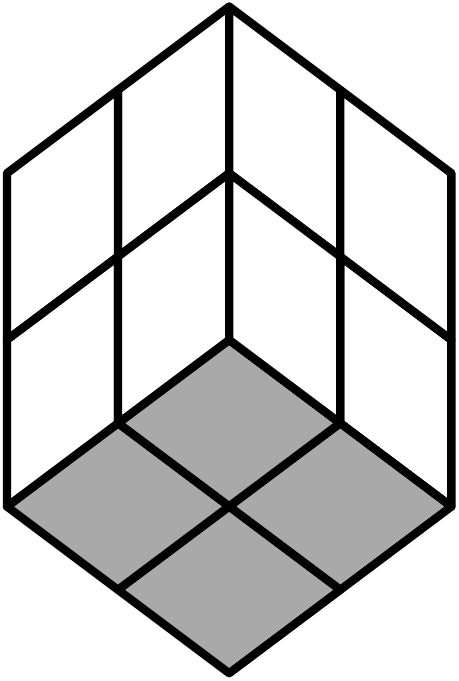}\hspace{.5in}\includegraphics[scale=.5]{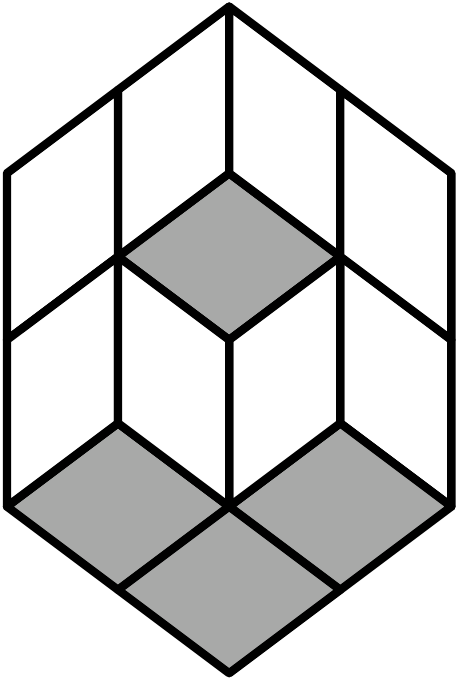}
\caption[Path Coupling of Monotonic Surfaces in Three Dimensions]{A pair of downsets $\sigma_t$ (left) and $\rho_t$ (right) where $\rho_t=\sigma_t\union\{(0,0,0)\}$.  These downsets differ on $x=(0,0,0)$, where $\mmon$ can increase $\phi_t$ by adding $x+\ubar_i,$ for any $i$.}
\label{3dCoupProbFig}
\end{figure}
}

The following bounds the mixing rate of the chain $\mmon$ when the bias is large enough.

%Next we see how to use the exponential distance metric to bound the mixing time of $\mmon$, showing:
\begin{theorem}
\label{FixedTilThm}
Let $R$ be any nice $d$-dimensional region with volume $n$, span $\alpha$, stretch $\gamma$, and uniform bias $\lambda$.  Let $\widehat{\lambda}=(d+\sqrt{d^2-4})/{2}$.  Define $\chi =  \widehat{\lambda}-\lambda^{-1/2}$.
\begin{enumerate}
\item If $\chi> 0$, then the mixing time of $\mmon$ satisfies $\tau(\varepsilon)=O\left(\chi^{-2}\alpha (\gamma\ln \lambda + \ln n) \ln \varepsilon^{-1}\right).$
\item If $\chi\geq 0$, then $\tau(\varepsilon)=O( n^2\alpha(\gamma\ln\lambda +\ln n)^2 \ln \epsilon^{-1})
.$
\end{enumerate}
In particular, if $\lambda\geq d^2$ then $\chi\geq 1/d^3$, so $\tau(\varepsilon)=O\left(d^3\alpha (\gamma\ln \lambda + \ln n) \ln \varepsilon^{-1}\right).$  
\end{theorem}
\noindent  Note that for all nice regions, $\alpha, \gamma < n$, so the mixing time of $\mmon$ is always polynomially bounded for the given biases.  
When $R$ is an $h_1 \times h_2 \times \dots \times h_d$ hyper-rectangular region, we get optimal bounds as long as each of the $h_i$ are fairly close.  Let $h_{min}=\min_i h_i$.  The span of $R$ is $\alpha=O({n}/{h_{min}})$ and the stretch is $\gamma=O( h_{min})$.  Hence we get $\tau(\varepsilon)=O( n + \frac{n\ln n}{h_{min}} \ln \varepsilon^{-1}),$ which is optimal as long as $h_{min}=\Omega(\ln n).$  In particular, we get optimal bounds when the region is a hypercube with $\alpha = O(h^{d-1})$ and $\gamma = O(h).$
%, proving Theorem~\ref{HigherDSimple}.

Next we define the exponential metric we need to prove Theorem \ref{FixedTilThm}.  Let $x_0$ be any vector in $R$ with maximal $L_1$ norm.  Then define
$$\phi(\sigma,\rho)=\sum_{x\in\sigma\oplus\rho}(\sqrt{\lambda})^{\|x_0\|_1-\|x\|_1}.$$
  We return to the coupling of $(\sigma_t,\rho_t)$ that simply supplies the same $(v^*, b, p)$ to both $\sigma_t$ and $\rho_t$. We let $U$ be the set of downsets that differ on a single cube. However, instead of the Hamming distance, we use the distance metric
$$\phi(\sigma,\rho)=\sum_{x\in\sigma\oplus\rho}(\sqrt{\lambda})^{\|x_0\|_1-\|x\|_1},$$
where $x_0$ is any vector in $R$ with maximal $L_1$ norm.  
We will show that this distance metric satisfies non-negative contraction in $\phi_t$, which is one of the requirements for Theorem \ref{GeomThm}. However, before we can prove that the distances decrease on average, we examine the moves which can \emph{increase} the distance.

For a pair $(\sigma_t,\rho_t)\in U$, there are two different ways the distance can increase in $(\sigma_{t+1},\rho_{t+1})$. If $\sigma_t=\rho_t\union\{x\}$, we can increase the distance by attempting to add a $v$ that succeeds in $\sigma_t$ but fails in $\rho_t$. This occurs when $v=x+\ubar_i$ for some $i$, so $v$ is ``supported" in $\sigma_t$ but not $\rho_t$. The other way to increase the distance between $\sigma_t$ and $\rho_t$ is to remove a $v$ that succeeds in $\rho_t$ but not in $\sigma_t$. This occurs when $v=x-\ubar_i$ for some $i$, as the move creates a valid downset in $\rho$ but not in $\sigma$. The following lemma bounds the number of such increases in distance.

\begin{lemma}
\label{BadBoundLemma}
For $\sigma_t=\rho_t\union\{x\}$, there are at most $d$ choices of $(v^*,b)$ such that $\phi_t$ increases.
\end{lemma}

\begin{proof}

We prove the Lemma by claiming that for dimensions $i\not= j$, if $\mmon$ can increase the distance by choosing $v=x+\ubar_i$, then it cannot increase the distance by choosing $v=x-\ubar_j$.
This follows from a proof by contradiction: If $\mmon$ can increase the distance with $x+\ubar_i$, then it is because $\rho_{t+1}=\rho_t\union\{x+\ubar_i\}$ is a valid downset. That means $x+\ubar_i-\ubar_j\in \rho_t$. On the other hand, if $\mmon$ can increase the distance with $x-\ubar_j$, it is because $\sigma_{t+1}=\sigma_t\backslash\{x-\ubar_j\}$ is a valid downset, which is only true if $x-\ubar_j+\ubar_i\not\in \sigma_t$. But this contradicts the fact $\sigma_t\oplus \rho_t=\{x\}$, justifying our claim.

This implies that, to increase the distance, $\mmon$ may add vectors of the form $x+\ubar_i$ for various dimensions $i$ as in Figure \ref{3dCoupProbFig}, or it may remove vectors of the form $x-\ubar_i$ for various dimensions $i$,
%as in Figure \ref{PathCoupRemFig},
 or it may add $x+\ubar_i$ and remove $x-\ubar_i$ in a single dimension $i$, as in Figure \ref{PathCoupBothFig}. In each of these cases, there are at most $d$ choices of $v$ that increase the distance.
\end{proof}

\condcomment{\boolean{includefigs}}{ 
\begin{figure}[ht]
\centering
\includegraphics[scale=.5]{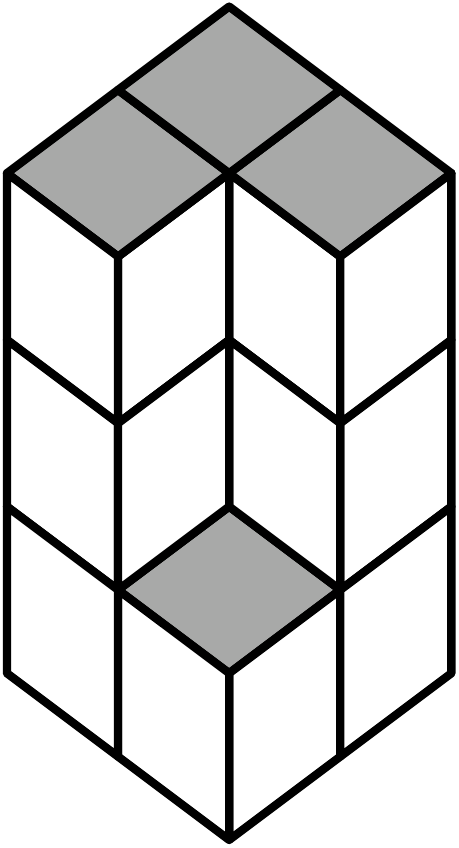}
\hspace{.5in}
\includegraphics[scale=.5]{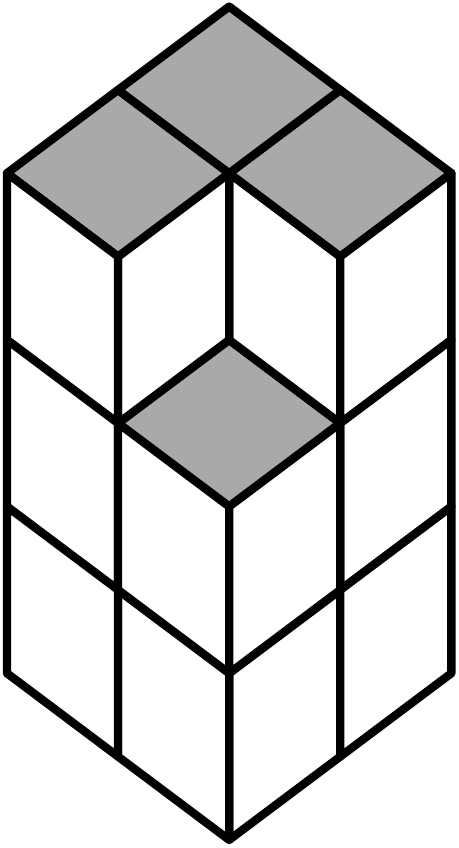}
\caption{Downsets that differ on $x$, where $\mmon$ increases $\phi_t$ by adding the vector above $x$ or removing the vector below $x$.}
\label{PathCoupBothFig}
\end{figure}
}
We may now prove Theorem \ref{FixedTilThm}.

%\begin{proof}
\vskip.2in
\noindent {\it Proof of Theorem \ref{FixedTilThm}}. \ 
We will show that the distance metric $\phi$ defined above satisfies the conditions of Path Coupling Theorem for Exponential metrics, Theorem \ref{GeomThm}.  First we want to show that the expected change in $\phi$ is negative.  By Lemma~\ref{BadBoundLemma}, there are at most $d$ choices of $(v^*,b)$ that can increase $\phi_t$.  We claim that each of these has an expected increase of at most $\phi_t\lambda^{-1/2}$.  To see this, consider a move of the form $v=x+\ubar_i$ for some $i$.  Then the increase in distance is $(\sqrt{\lambda})^{\|x_0\|_1 - \|v\|_1}={\phi_t}{\lambda}^{-1/2}$.  If the move is of the form $v=x-\ubar_i$ for some $i$, then the increase in distance is ${\lambda}^{-\|v\|_1/2}=\phi_t \sqrt{\lambda}$, but the chance of choosing an appropriate $p$ is ${1}/{\lambda}$. Therefore the expected increase is at most ${\phi_t}{\lambda}^{-1/2}$.

There are also two choices of $(v^*,b)$ that \emph{decrease} $\phi_t$; corresponding to adding~$x$ and removing~$x$. These each decrease $\phi_t$ by $\phi_t$, and succeed with probability $1$ and ${1}/{\lambda}$, respectively. Therefore the expected change in distance satisfies
\begin{align*}
\E[\phi_{t+1}-\phi_t] &\leq \frac{1}{2\alpha} \left(d\cdot\frac{\phi_t}{\sqrt{\lambda}}-\left(1+\frac{1}{\lambda} \right)\phi_t\right)\\
	&= \frac{\phi_t}{2\alpha} \left(\frac{d}{\sqrt{\lambda}}-1-\frac{1}{\lambda} \right)\\
	&= -\frac{\phi_t}{2\alpha} \left(\frac{1}{\sqrt{\lambda}} - \frac{d+\sqrt{d^2-4}}{2}\right) \left(\frac{1}{\sqrt{\lambda}} - \frac{d-\sqrt{d^2-4}}{2}\right)\\
	&\leq - \frac{\phi_t\chi^2}{2\alpha}.
\end{align*}

Next we check the other conditions of Theorem~\ref{GeomThm}.
For arbitrary $\sigma,\rho\in\Omega_{mon}$, if $x\in\sigma\oplus\rho$ for some $x$, then $\phi(\sigma,\rho)\geq \sqrt{\lambda}^{\|x_0\|_1-\|x\|_1}\geq 1$. Therefore if $\phi(\sigma,\rho)<1$, $\phi(\sigma,\rho)=0$.  Let $U$ be the set of pairs of downsets that differ on a single vector. For arbitrary $\sigma,\rho\in\Omega_{mon}$, we can connect $\sigma$ to $\rho$ by simply adding or removing the vectors in $\sigma\oplus\rho$ one by one, and $\phi(\sigma,\rho)$ is the sum of the distances.  Since the volume of $\widehat{R}$ is $n$, there are at most $n$ possible vectors in $\sigma\oplus\rho$, so $\phi(\sigma,\rho)\leq n\lambda^{\gamma/2}$ for all~$\sigma,\rho$.  

We consider two cases.  If $\chi>0$, then $\E[\phi_{t+1}]\leq \beta \phi_t$, where $\beta=1-\chi^2/(2\alpha)$.  Thus, by Theorem \ref{GeomThm}, we have $\tau(\epsilon)=O(\chi^{-2}\alpha (\gamma\ln\lambda + \ln n)\ln \epsilon^{-1})$.  
On the other hand, if $\chi \geq 0,$ but less than any constant, then we use the second part of Theorem ~\ref{GeomThm}.  For any pair of $\sigma,\rho$, $\mmon$ can always add a vector $v^*$ in their difference that maximizes $\|v\|_1$.  This would change $\phi_t$ by $(\sqrt{\lambda})^{\alpha-\|v^*\|_1}$.  On the other hand, $\phi_t\leq n(\sqrt{\lambda})^{\alpha-\|v^*\|_1}$ so the change in $\phi_t$ is at least $\phi_t/n$.  
  The appropriate $v^*$ is chosen with probability at least ${1}/{\alpha}$ and the appropriate $b$ is chosen with probability ${1}/{2}$ (and every $p$ succeeds when adding).  Therefore there is a ${1}/{(2\alpha)}$ chance of changing $\phi_t$ by ${\phi_t}/{(2n)}$, i.e. $P\left(|\phi_{t+1}-\phi_t|\geq {\phi_t}/{n}\right)\geq {1}/{(2\alpha)}.$
Hence in this case
$$\tau(\epsilon)=O\left( \frac{\ln^2(B)}{\ln^2(1+\frac{1}{n})\frac{1}{2\alpha}}  \ln \epsilon^{-1}\right)=O( n^2\alpha(\gamma\ln\lambda +\ln n)^2 \ln \epsilon^{-1}).$$
\hfill \square
\vskip.2in
%\end{proof}

\subsection{Hitting time to the maximal tiling}\label{ceiling}

We now introduce a second technique that allows us to get improved bounds for the mixing
rate of the uniform bias Markov chain $\mtil$ whenever $\lambda\geq d$ and the region is an $h\times h\times \cdots h$ hypercube of volume $n=h^d$.  Specifically, we prove:

\begin{theorem}\label{NewCeilThm}
Let $R\subset \Z^d$ be the $h\times h\times\cdots\times h$ hypercube of volume $n=h^d$ and bias $\lambda\geq d$.  Then the mixing time of $\mtil$ satisfies $\tau(\varepsilon)=O(h^{2d-1}\ln\varepsilon^{-1}).$  In general, this is $o(n^2\ln\varepsilon^{-1}),$ or in 2 dimensions, $O(n\sqrt{n}\ln\varepsilon^{-1}).$
\end{theorem}
\noindent 
The proof relies on the monotonicity of $\mmon$ with respect to the trivial coupling.  In other words, if $(X_t, Y_t)$ are coupled and $X_t \subseteq Y_t$, then after one step of the coupling, $X_{t+1}\subseteq Y_{t+1}$.  This implies that the coupling time is bounded by the time to hit the full cube $F$ starting from the empty cube~$F'$, and we can show that this will happen quickly because the distance to $F$ is always non-increasing in expectation (see Fig.~\ref{CeilFig}b).

\condcomment{\boolean{includefigs}}{ 
\begin{figure}[ht]
\centering
\includegraphics[scale=.07]{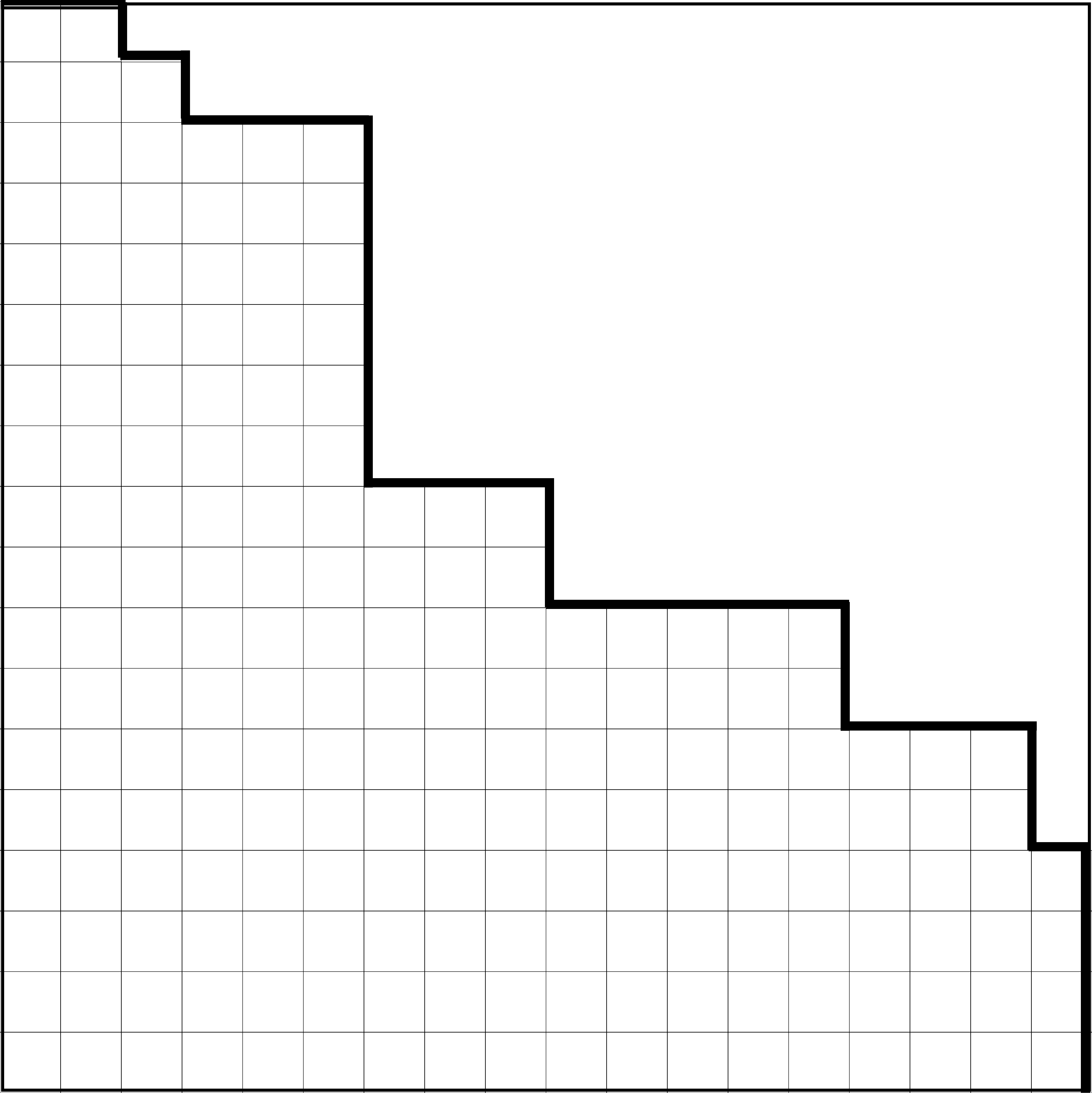}
\put(-60,60){$X$}
\put(-114,8.1){$F'$}
\put(-20,105){$F$}
\hspace{1in}
\includegraphics[scale=.07]{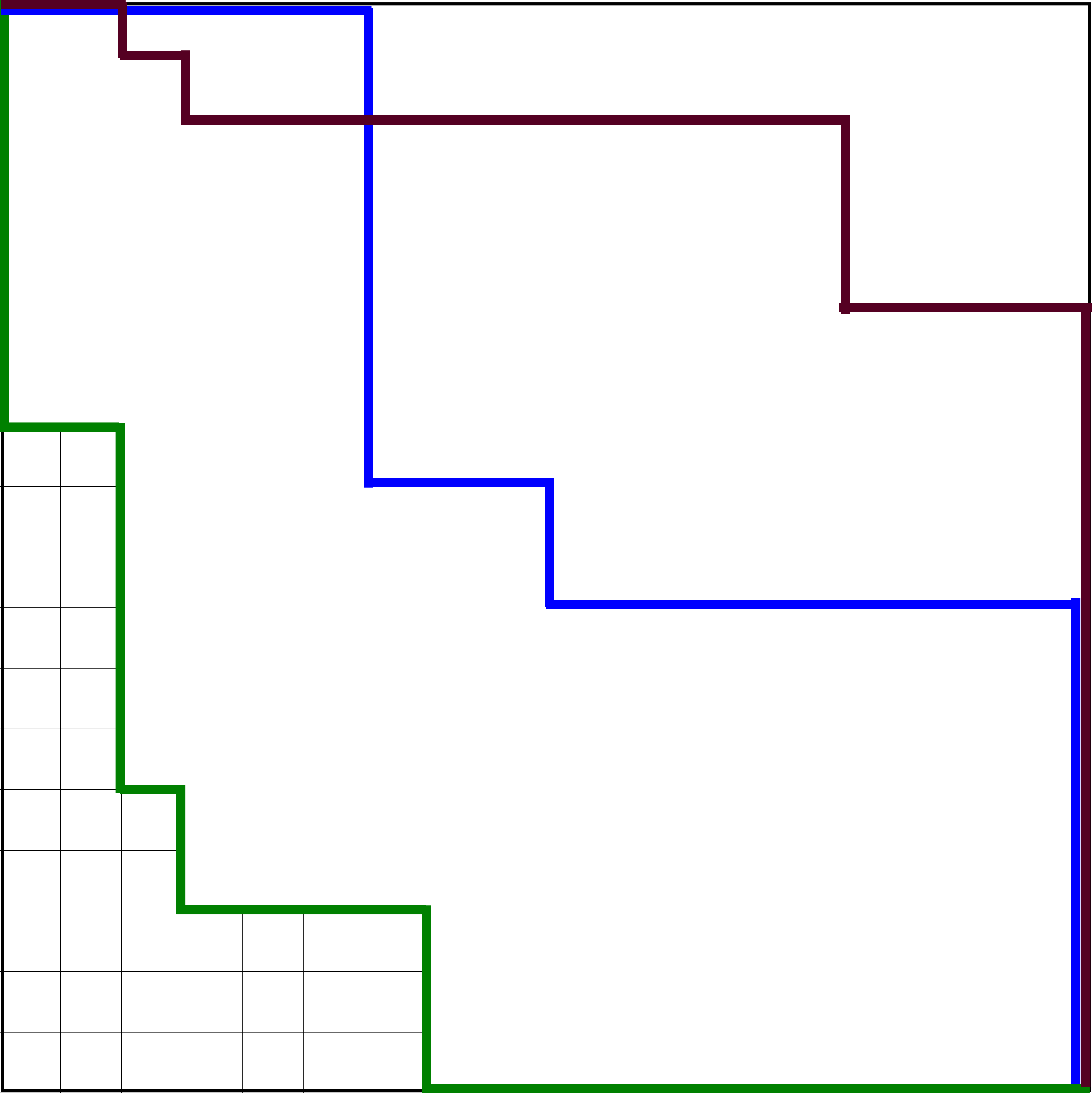}
\put(-80,60){$Y_t$}
\put(-38,78){$X_t$}
\put(-75,4.1){$F_t$}
\caption[Coupling with the maximal tiling]{ (a)  A staircase walk with $5$ peaks and $4$ valleys.  (b) $F_t\subseteq X_t$ and $F_t\subseteq Y_t$ for all $t $}
\label{CeilFig}
\end{figure}
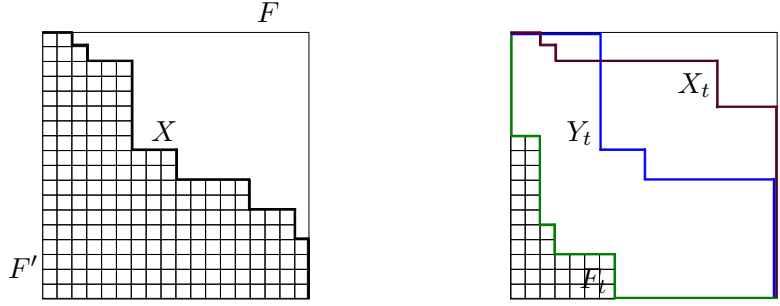
}

We begin by proving that a monotone Markov chain with nonnegative drift will mix rapidly.  Call a Markov chain $\cM$ on $\Omega$ \emph{monotone} if $\Omega$ is a distributive lattice with partial order $\pleq$ and there exists a coupling $(X_t,Y_t)$ such that if $X_t\pleq Y_t$, then $X_{t+1}\pleq Y_{t+1}$.
For $X_t\in\Omega$, let $\p_\prec(X_t)=\p[X_{t+1}\not=X_t \text{~and~} X_{t+1}\pleq X_t]$ and 
let $\p_\succ(X_t)=\p[X_{t+1}\not=X_t \text{~and~} X_t\pleq X_{t+1}]$.

\begin{theorem}
\label{CeilingThm}
Given a monotone Markov chain $\cM$ on $\Omega$ with maximal element $F$ and minimal element $F'$, let $h(X,Y)$ be the Hamming distance between $X$ and $Y$ and let $H=\max_{X,Y\in \Omega} \{h(X,Y)\}$. Assume that for all $X\not=F\in\Omega$,
$$\p_\succ(X)-\p_\prec(X) \geq \kappa\geq 0.$$
\begin{enumerate}
\item If $\kappa>0$ then the mixing time is $\tau(\epsilon)\leq\lceil \frac{eH\ln H}{\kappa}\rceil \lceil \ln(\epsilon^{-1})\rceil.$
\item If $\kappa=0$ and for all $t\geq 0$, $\p[X'\neq X]\geq Q$ then the mixing time satisfies $\tau(\epsilon) \leq \lceil \frac{2eH^2}{Q}\rceil \lceil \ln(\epsilon^{-1})\rceil.$
\end{enumerate}
\end{theorem}

\begin{proof}
First we notice that the coupling time $T^{x,y} = \E[\min\{t:  X_t=Y_t | X_0=x, Y_0=y\}]$ is bounded by the hitting time to reach $F$ from $F'.$  Define $T^{x}_F=\E[\min\{t: X_t=F | X_0=x\}]$ to be the hitting time to reach $F$ from $x$.  Let $F_0=F'$, and couple the moves of $\{F_t\}$ together with the moves of $\{X_t\}$ and $\{Y_t\}$; that is, choose the same $(i,b,r)$ for $F_t$ as in $X_t$ and $Y_t$.  Since the Markov chain is monotone with respect to the given coupling, we have $F_t\pleq X_t$ and $F_t\pleq Y_t$ for all $t\geq 0$.  Thus if $F_t=F$, we also have $X_t=Y_t=F$.  So the coupling time for $X_t$ and $Y_t$ is bounded by the hitting time of $F_t$ to $F$.  See Figure~\ref{CeilFig}(b).

Let $\phi_t=h(X_t,F)$.  For part (1), we use $\E[\phi_{t}]\leq\phi_t-\kappa\phi_{t-1}/H\leq(1-\kappa/H)^{t}\phi_0$.  For part (2), note that $\E[\phi_{t+1}^2+\phi_t^2]\geq \phi_t^2\geq 1$, so using Lemma~\ref{MartingaleLemma} with $q=0$ and $D=H$, we obtain $\E[T]\leq \frac{2H^2}{Q}.$
Then the theorem follows from the Coupling Lemma, Theorem~\ref{couplinglemma}.
\end{proof}

Notice that since $\lambda>1$, we expect that for any downset $X_0\in \Omega$, the sequence $\{X_t\}$ should approach the unique maximal element of $\Omega$. We will show that it suffices to choose $\lambda\geq d$ to reach the full $d-$dimensional cube from an arbitrary position in polynomial time, thus achieving polynomial mixing time. In particular, we will show that $\mtil$ satisfies the conditions of Theorem~$\ref{CeilingThm}$.

We define a {\it peak} of $X$ to be a position where we can remove a hypercube from $X$ and a {\it valley} of $X$ to be a position where we can add a hypercube to $X$ (see Figure~\ref{CeilFig}(a)).  For the following lemma,  define $\cV(D)$(and, respectively, $\cP(D)$) to be the set of valleys (peaks) of a downset $D$.

\begin{lemma}
\label{LozCeilingLemma}
For any downset $\sigma \neq F \in \Omega$, 
\begin{equation}
|\cP(\sigma)|\leq (d-1) |\cV(\sigma)| +1. 
\end{equation}
\end{lemma}

\begin{Proofof}{Lemma~\ref{LozCeilingLemma}}
First we notice that if $\sigma\neq F$, then it has at least one valley.  Furthermore, if $\sigma$ has a single valley, then the number of peaks is at most $d=(d-1)\cV(\sigma) +1$.  Now, assume the number of valleys is more than one and proceed as follows.  Choose a valley $v=(v_1,v_2,\ldots, v_d)$ that maximizes $v_d$.  Construct $\sigma'$ from $\sigma$ by adding every hypercube lying above $v$ in the $d$ dimension; that is, $\sigma' = \sigma \cup \{v+e_d, v+2e_d, \ldots, v+(h-v_d)e_d\}$.  Then $|\cP(\sigma')|\geq |\cP(\sigma)|-d+1$, and $|\cV(\sigma')|=|\cV(\sigma)|-1$.  Hence
$$|\cP(\sigma)|-(d-1) |\cV(\sigma)|\leq  |\cP(\sigma')| - (d-1)(|\cV(\sigma')|.$$
Obtain $\widehat{\sigma}$ by iterating this operation until only a single valley remains; then we have
$$|\cP(\sigma)|- (d-1)|\cV(\sigma)|\leq |\cP(\widehat{\sigma})|-(d-1) |\cV(\widehat{\sigma})|\leq 1,$$
as desired.
\end{Proofof}

%%%%%%%%%%%%%%%%%%%%%%%%%%%%%%%
%%%                         New Improvement			%%%%%
%%%%%%%%%%%%%%%%%%%%%%%%%%%%%%%
Lemma~\ref{LozCeilingLemma} implies that if $\lambda\geq d$ then for all $ \sigma\not=F\in\Omega$, $|\cP(\sigma)|\leq (d-1) |\cV(\sigma)| +1\leq d|\cV(\sigma)| $, and so
\begin{eqnarray*}
\p_\succ(\sigma)-\p_\prec(\sigma)&\ = \ & \frac{1}{2\alpha}\left( |\cV(\sigma)|-\frac{1}{\lambda}|\cP(\sigma)| \right)\\
& \ \geq \ &  \frac{|\cV(\sigma)|}{2\alpha}\left(1-\frac{d}{\lambda}\right)\geq 0.
\end{eqnarray*}
Hence, the Markov chain $\mtil$ has nonnegative drift towards the maximal configuration.  Using Theorem~\ref{CeilingThm} we can show that for $\lambda=d+\delta$ for some $\delta\geq 0$,
$\tau(\varepsilon)= O(\delta^{-1} \alpha n\ln n\ln\varepsilon^{-1})$ for $\delta>0$, and otherwise $\tau(\varepsilon)= O( \alpha n^2 \ln\varepsilon^{-1})$.  However, we are able to get a better bound in Theorem~\ref{NewCeilThm} by introducing another metric that gives strict additive contraction for every configuration, even when $\lambda=d$.

\vspace{.2in}
\noindent {\em Proof of Theorem \ref{NewCeilThm}.}

Recall from the proof of Theorem~\ref{CeilingThm} that the coupling time is bounded by the hitting time to reach $F$ from $F'.$  We will show that the expected time to hit $F$ is small, using Lemma~\ref{LozCeilingLemma}.   Let  $S_1 = \{\sigma\in \Omega : \cP(\sigma)<d \cV(\sigma)\}\cup F,$ and $S_2 = \{\sigma\in \Omega : \cP(\sigma)=d \cV(\sigma)\}.$  Define the function $\phi(\sigma)= H(F,\sigma) + I_{S_2}/({2d})$, where $H(X,Y)$ is the Hamming distance between the downsets $X$ and $Y$, and $I_{S_2}$ is the indicator function for the set $S_2$.  We will show that $\phi$ has negative additive drift towards $0$.

Notice that Lemma~\ref{LozCeilingLemma} implies that if $\sigma\in S_2$ then it has exactly one valley.  Moreover, if $\sigma\in S_1$ can move to $S_2$ in a single step, then the number of possible moves to take it to $S_2$ is at most $3$.  If $\sigma$ has a single valley then it must have a valley $v$ such that for some dimension $i$, $v_i=0$; adding a cube at this valley could move $\sigma$ into $S_2$.  If $\sigma$ has three valleys then it has a single hypercube $c$ it can remove to enter $S_2$, and if $\sigma$ has two valleys then it can enter $S_2$ by adding cubes at either of those valleys or removing a cube between them are the only ways to potentially move into $S_2$.  Hence, if $\sigma_t\in S_1$, then the probability that $\sigma_{t+1}\in S_2$ is at most ${3}/({2\alpha})$.  Moreover, if $\sigma_t \in S_2$, then the probability that $\sigma_{t+1}\in S_1$ is at least ${1}/({2\alpha})$.  

Now, conditioning on whether $\sigma_t $ is in $ S_1$ or in $S_2$, we have $\E[\phi(\sigma_{t+1})-\phi(\sigma_t) | \sigma_t \in S_2] \leq 0 -{2d})^{-1}({2\alpha})^{-1}$
and
\begin{eqnarray*}
\E[\phi(\sigma_{t+1})-\phi(\sigma_t) | \sigma_t \in S_1] &\leq& \frac{-|\cV(\sigma_t)| + \lambda^{-1} |\cP(\sigma_t)|}{2 \alpha} +\frac{3}{2\alpha} \left(\frac{ 1}{2d} \right)\\
&\leq& \frac{1}{2\alpha}\left(-|\cV(\sigma_t)| + \lambda^{-1}(d-1)|\cV(\sigma_t)| \right)+\frac{3}{2\alpha} \left(\frac{ 1}{2d} \right)\\
&\leq& \frac{|\cV(\sigma_t)| }{2\alpha}\left(-\frac{1}{d}  \right)+\frac{3}{2\alpha} \left(\frac{ 1}{2d} \right)\\
&\leq& \frac{-1}{4\alpha d}.
\end{eqnarray*}
Thus 
$$\E[\phi(\sigma_{t})]\leq \sigma_0 - \frac{t}{4\alpha d}\leq \left(n+\frac{1}{2d}\right)- \frac{t}{4\alpha d}\leq \epsilon$$
whenever $t\geq 4d \alpha (n+\frac{1}{2d} - \epsilon)$.  Thus the mixing time satisfies $\tau(\varepsilon)O(\alpha n\ln(\varepsilon^{-1}) )=O(h^{2d-1}\ln(\varepsilon^{-1}) )= O(n^2\ln(\varepsilon^{-1}))$ for arbitrary dimension $d$, but for $d=2$, this implies $\tau(\varepsilon)=O(h^3\ln(\varepsilon^{-1}))=O(n\sqrt{n}\ln(\varepsilon^{-1})).$

\qed

\comment{

Unfortunately, this easy result for the Markov chain defined on all downsets $D$ contained in the region $\hat{R_d}$ does not extend to all regions. In other words, we cannot define any constant $\lambda$ such that the Markov chain $\mmon$ mixes rapidly on any arbitrary region.

\begin{theorem}
\label{LozCeilingNegThm}
Let $\Omega$ be the set of all downsets within region $R$. Then for any constant $\lambda,$ there exist regions $R$ and $X\subsetneq R$ such that $\p_\prec(X) \geq \p_\succ(X)$.
\end{theorem}

\begin{proof}
Let $R=\{\xbar: \sum_i x_i \leq l\}$, and $X=V\backslash \{\xbar*\}$, for some $\xbar^*$ s.t. $\sum_i x^*_i=l$.
Then $\cV(X)$ consists of a single vector: $\xbar^*$. However, $\cP(X)$ contains all $\xbar\neq \xbar^*$ such that $ \sum_i x_i = l$, and thus is exponentially large.
\end{proof}
}

\section{The fluctuating bias Markov chain}\label{FluctSection}
It turns out to be quite interesting to consider the case of {\it fluctuating bias} where the bias of a tile depends on its position.
This situation is more realistic, particularly in the context of self-assembly.  For instance, the probability of a tile lined with DNA attaching to the substrate depends on the strength of the bonds along the edges of the tile as well as the relative densities of each tile.  Recall the {\it bias at $\bar{x}$} is defined as follows: if $\tau$ is formed by adding a cube at position $\bar{x}$ to $\sigma$, then $\lambda_{\bar{x}} = P(\sigma, \tau)/P(\tau, \sigma)$ is called the {\it bias at $\bar{x}$}.  The stationary probability of a configuration $\sigma$ will be proportional to $\prod_{\bar{x}\in \sigma}\lambda_{\bar{x}}$.%, where by $\bar{x}\in \sigma$ we mean that $\bar{x}$ is a cube lying below the surface $\sigma$. 

Most of the results from Section~\ref{UnifSection} generalize to this setting, as long as we satisfy certain bounds on the amount the bias can fluctuate.  When the minimum bias $\lambda_L$ is large enough, we do not need any upper bound.  While the upper bound might seem unnecessary even for small values of $\lambda_L> 1$, the chain can actually take exponential time to reach equilibrium if the biases vary too much (see Section~\ref{SlowSection}).

\subsection{Fast mixing with large enough minimum bias}\label{bounded}
In Section~\ref{ceiling}, we showed that: (i) the coupling time for the Markov chain $\mtil$ was bounded by the hitting time to the maximal configuration, and (ii) the hitting time is polynomial, assuming that the bias is large enough.  Clearly, (i) still holds for the fluctuating bias Markov chain $\mfluct$, and the hitting time to the maximal configuration should only decrease if some of the biases are increased.  Therefore, the following theorem is a simple consequence of Theorem~\ref{NewCeilThm}:

\begin{theorem}\label{HigherDFluctSimple}
Let $R$ be the d-dimensional $h\times h\times\cdots\times h$ hypercube with volume $h^d=n$ and fluctuating bias.  Assume the minimum bias $\lambda_L$ satisfies $\lambda_L\geq d$.   Then 
the mixing time of $\mmon$ satisfies $\tau(\varepsilon)=O\left(h^{2d-1}\ln \varepsilon^{-1}\right)=o(n^2\ln \varepsilon^{-1}).$
\end{theorem}

Moreover, we can use a similar argument, together with the uniform bias results above, to obtain the following stronger result for fluctuating bias in 2 dimensions:
\begin{theorem}\label{2DFluctSimple}
Let $R$ be a rectangular $h\times w$ region in $\Z^2$ with fluctuating bias.  Suppose the minimum bias  $\lambda_L$ is a constant larger than 1.   Suppose without loss of generality that $h\leq w$. Then the mixing time of $\mmon$ satisfies $\tau(\varepsilon)=O\left((w+h) (h + \ln w) \ln \varepsilon^{-1}\right).$
\end{theorem}
\noindent Again, this yields the optimal mixing time of $O(n)$ if $R$ is a square.

\vspace{.2in}
\noindent {\em Proof of Theorem \ref{2DFluctSimple}.}
Recall from Theorem~\ref{CeilingThm} that the hitting time to the
maximal configuration is an upper bound on the coupling time.  Thus,
it suffices to show that for uniform bias $\lambda$, where $\lambda>1$
is a constant, the hitting time to the maximal configuration is
$O\left((w+h) (h + \ln w)\ln \varepsilon^{-1}\right).$  For square
regions with area $n$, a bound of $O(n)$ was proved by Benjamini et
al.~\cite{BBHM05}.  We wish to extend this to rectangular regions. 

It is well-known that in 2 dimensions for uniform bias $\lambda$,
where $\lambda>1$ is a constant, the maximal configuration has
constant probability in the stationary distribution.  For
completeness, we prove this next in the fluctuating bias setting. 

Let $p(h,w;t)$ denote the number of
integer partitions of $t$ into at most $h$ parts, each of size at most
$w$.  This is precisely the number of staircase walks in an $h\times
w$ region with $t$ squares below the curve.  Consider the generating
function for $p(h,w;t)$:
$$F(q) = \sum_{t=0}^{hw}p(h,w;t)q^t = \binom{h+w}{w}_q,$$
where $\binom{m}{r}_q= \prod_{i=0}^{r} \frac{1-q^{m-i}}{1-q^i}$ is the
Gaussian binomial coefficient.  Then the normalizing constant $Z$ is
equal to $F(\lambda)$.
We wish to show that the weight of the highest configuration
$\lambda^{hw}/Z$ is at least a constant, independent of $h$ and $w$.

By rearranging terms, we have 
$$F(q) =  \prod_{i=1}^{w}\frac{q^{h+i}-1}{q^i-1}\leq \prod_{i=1}^{w} \frac{q^{h+i}}{q^i-1}= q^{hw}\prod_{i=1}^{w} \frac{1}{1-q^{-i}}.$$
Let $x=1/\lambda.$  Then $F(\lambda) =
q^{hw}\prod_{i=1}^w\frac{1}{1-x^i},$ and so
\begin{align}
\ln(Z\lambda^{-hw})\leq \ln(F(\lambda)\lambda^{-hw})&= -\sum_{t=1}^{w} \ln(1-x^{t})\nonumber\\
&= \sum_{t=1}^{hw} x^t+\frac{x^{2t}}{2}+\frac{x^{3t}}{3}+\ldots \nonumber\\
&\leq  x + 2x^2+ 3x^3+4x^4+\ldots\label{rearrange}\\
&= \frac{x}{(1-x)^2}.\nonumber
\end{align}
Inequality~\ref{rearrange} follows because $x^i$ appears in at most
$i$ terms of the sum, each with a coefficient at most 1.
Therefore the maximum configuration has weight $\lambda^{hw}/Z\geq
e^{-x/(1-x)^2}$, for any $h$ and $w$.

By Theorem~\ref{2DSimple}, we know that the mixing time of the uniform
bias chain with bias $\lambda_L$ is $O\left(w (h + \ln w) \ln
\varepsilon^{-1}\right)$, so we expect the uniform bias chain to hit
the maximal configuration in $O\left(w (h + \ln w) \ln
\varepsilon^{-1}\right)$ steps.   

The hitting time of the fluctuating bias Markov chain is at most the hitting time of the uniform bias chain whose bias is equal to $\lambda_L$.  To see this, we will couple the two chains.  We start with the uniform bias chain below the fluctuating chain, for example by setting it equal to the minimal configuration.  At each step, we choose the same square to add, and whenever the fluctuating bias chain decides to remove a square, the uniform chain does too; this is possible because the bias of the fluctuating chain is at least the bias of the uniform bias chain at every square.  Therefore, monotonicity is preserved during the coupling, so when the uniform bias chain hits the maximal configuration, so must the fluctuating bias chain.
\qed

\subsection{Fast mixing when the fluctuations are bounded}\label{close}
We can also extend the exponential metric technique of Section~\ref{UnifSection} to handle fluctuating bias, provided the biases $\lambda_x$ for $x\in R$ do not vary too much.
\begin{theorem}
\label{FluctTilThm}
Let $R$ be any nice $d$-dimensional region with volume $n$, span $\alpha$, stretch $\gamma$, and suppose the bias at any point $x$ satisfies $1<\lambda_L\leq \lambda_x\leq \lambda_U$.  If the maximum and minimum biases are such that 
\begin{equation}\label{ContCond}
\frac{d}{\sqrt{\lambda_L}}-1-\frac{1}{\lambda_U}\leq -\chi
\end{equation}
for some $\chi\geq 0$, then 
%Let $\widehat{\lambda}_d = \left(\frac{2}{d-\sqrt{d^2-4}}\right)^2$.
%If $d=2$, then let $\lambda >1$ be the bias. If $d \geq 3$, let $\lambda \geq d^2$ be the bias.

\begin{enumerate}
\item If $\chi> 0$, then the mixing time of $\mfluct$ satisfies $\tau(\varepsilon)=O\left(\chi^{-1}\alpha (\gamma\ln \lambda + \ln n) \ln \varepsilon^{-1}\right).$
%\item If $d=2$ and $0<\lambda_L -1=O(\frac{1}{\gamma})$, then $\tau(\varepsilon)=O(\alpha n^2\ln(\varepsilon^{-1})).$
\item If $\chi\geq 0$, then $\tau(\varepsilon)=O \left( \alpha n^2(\gamma\ln \lambda + \ln n)^2)\ln(\varepsilon^{-1})\right).$
\end{enumerate}
\end{theorem}

\begin{proof}
This theorem is proved nearly identically to Theorem~\ref{FixedTilThm}.  In this case we define the distance metric as
$$\phi(\sigma,\rho)=\sum_{x\in\sigma\oplus\rho}(\sqrt{\lambda_L})^{\|x_0\|_1-\|x\|_1}.$$
Given that $\mfluct$ chose $v^*$ and $b$ such that $\phi_{t}$ can increase, the expected increase is at most ${\phi_t}{\lambda_L}^{-1/2}$.  Indeed, if the move is of the form $v=x+\ubar_i$ for some $i$, then the increase in distance is $(\sqrt{\lambda_L})^{\|x_0\|_1 - \|v\|_1}={\phi_t}{\lambda_L}^{-1/2}$.
If the move is of the form $v=x-\ubar_i$ for some $i$, then the increase in distance is $\lambda_L^{-\|v\|_1/2}=\phi_t \sqrt{\lambda_L}$, but the chance of choosing an appropriate $p$ is ${1}/{\lambda_v}\leq {1}/{\lambda_L}$. Therefore the expected increase is again at most ${\phi_t}{\lambda_L}^{-1/2}$.

This implies that the expected change in distance is negative.  
%If $\frac{d}{\sqrt{\lambda_L}}-1-\frac{1}{\lambda_U}\leq -\chi$ for some $\chi\geq 0$ then for $\phi_t=\phi(\sigma_t,\rho_t)$,
%$$\E[\phi_{t+1}-\phi_t]\leq- \frac{\phi_t\chi}{2\alpha}.$$
As before, there are at most $d$ bad moves, but now the two good moves happen with probability $1$ and ${1}/{\lambda_x}\geq {1}/{\lambda_U}$, respectively. Therefore the expected change in distance satisfies
%\begin{align*}
$$\E_t[\phi_{t+1}-\phi_t] \leq \frac{1}{2\alpha} \left(d\cdot\frac{\phi_t}{\sqrt{\lambda_L}}-\left(1+\frac{1}{\lambda_U} \right)\phi_t\right) \ \leq \  - \frac{\phi_t\chi}{2\alpha}.$$
%	&= \frac{\phi_t}{2\alpha} \left(\frac{d}{\sqrt{\lambda_L}}-1-\frac{1}{\lambda_U} \right)\\
%	&\leq - \frac{\phi_t\chi}{2\alpha}.
%	&= -\frac{\phi_t}{2\alpha} \left(\frac{1}{\sqrt{\lambda}} - \frac{d+\sqrt{d^2-4}}{2}\right) \left(\frac{1}{\sqrt{\lambda}} - \frac{d-\sqrt{d^2-4}}{2}\right)\\
%\end{align*}
The rest of the proof is identical to the proof of Theorem~\ref{FixedTilThm}.
\end{proof}

\subsection{Slow mixing when the fluctuations are unbounded}\label{SlowSection}
Note that we have restricted throughout this chapter to the case when $\lambda_{x,y} \geq 1$ for all $(x,y)$.  This restriction is necessary, or the chain might not be rapidly mixing.  For example, it can be shown that if $\lambda_{x,y} <1$ when $x+y \leq n$ and $\lambda_{x,y} > 1$ when $x+y > n$, then $\mfluct$ requires exponential time to converge.  Indeed it will be difficult to 
move from a tiling that is nearly empty to one that is nearly full, even though these each occupy 
a constant fraction of the stationary probability.  In fact, we will see presently that even if the minimum bias satisfies $\lambda_L> 1$, the Markov chain may have exponential mixing time.
%Begin Slow Result!

Let $R$ be an $n\times n$ square region in 2 dimensions.  We consider the following question.  Suppose that for all $(x,y)\in R$, $\lambda_{x,y}\geq \lambda_L>1.$  We know by Theorem~\ref{FixedTilThm} that for uniform bias, the Markov chain is rapidly mixing for all $\lambda>1$ polynomially bounded away from $1$ and the result is in general easier to show for larger values of $\lambda$.  This leads us to expect rapid mixing in the fluctuating bias case, as long as the minimum bias $\lambda_L$ is polynomially bounded away from $1$.  However, this is not true in general:
\begin{theorem}\label{SlowSimple}
There exists a set of biases $\{\lambda_x\}$ on a square region of $\Z^2$ where $\lambda_x>1$ for all $x$ and yet the mixing time of $\mmon$ is $ \Omega(e^{n/24}).$
\end{theorem}
We show that if the biases below the line $x+y=n+M$ (where $M=n-\sqrt{n}$) are all close to 1 and all other biases are some very large constant $\xi$, then the mixing time of $\mmon$ is exponentially large in $n$ (see Figure~\ref{fluctuatingbias}).  We identify sets $S_1,S_2,S_3$ such that $\pi(S_2)$ is
exponentially smaller than both $\pi(S_1)$ and $\pi(S_3)$, which have equal weight, but to get between
$S_1$ and $S_3$, $\mnn$ and $\mt$ must pass through $S_2$, the cut.   This bad cut prevents the Markov chain from mixing rapidly, regardless of the initial configuration.

  To formalize these ideas, we will bound the conductance of the Markov chain.  The
 \emph{conductance} of an ergodic Markov chain~$\m$ with stationary distribution $\pi$ is
$$\Phi_\m = \min_{{S\subseteq\Omega}\atop{ \pi(S) \leq 1/2}} \phi_S,$$
where $\phi_S=\phi_S^{(\m)}$ is the {\em conductance of a set} $S\subset\Omega$, defined by
$$\phi_S = \frac{1}{\pi(S)}\sum_{s_1\in S, s_2\in \bar{S}}\pi(s_1)P(s_1,s_2).$$

\noindent Essentially, $S$ is a bad cut if $\phi_S$ is exponentially small.  The existence of such a set $S$ prevents the Markov chain from mixing rapidly~\cite{JS89}:
 \begin{theorem}
   \label{conductance}
   For any Markov chain with conductance $\Phi$, 
   \ $\tau \geq ({4\Phi})^{-1} - 1/2.$
 \end{theorem}

\begin{figure}[htb]
\centering
\includegraphics[scale=.08]{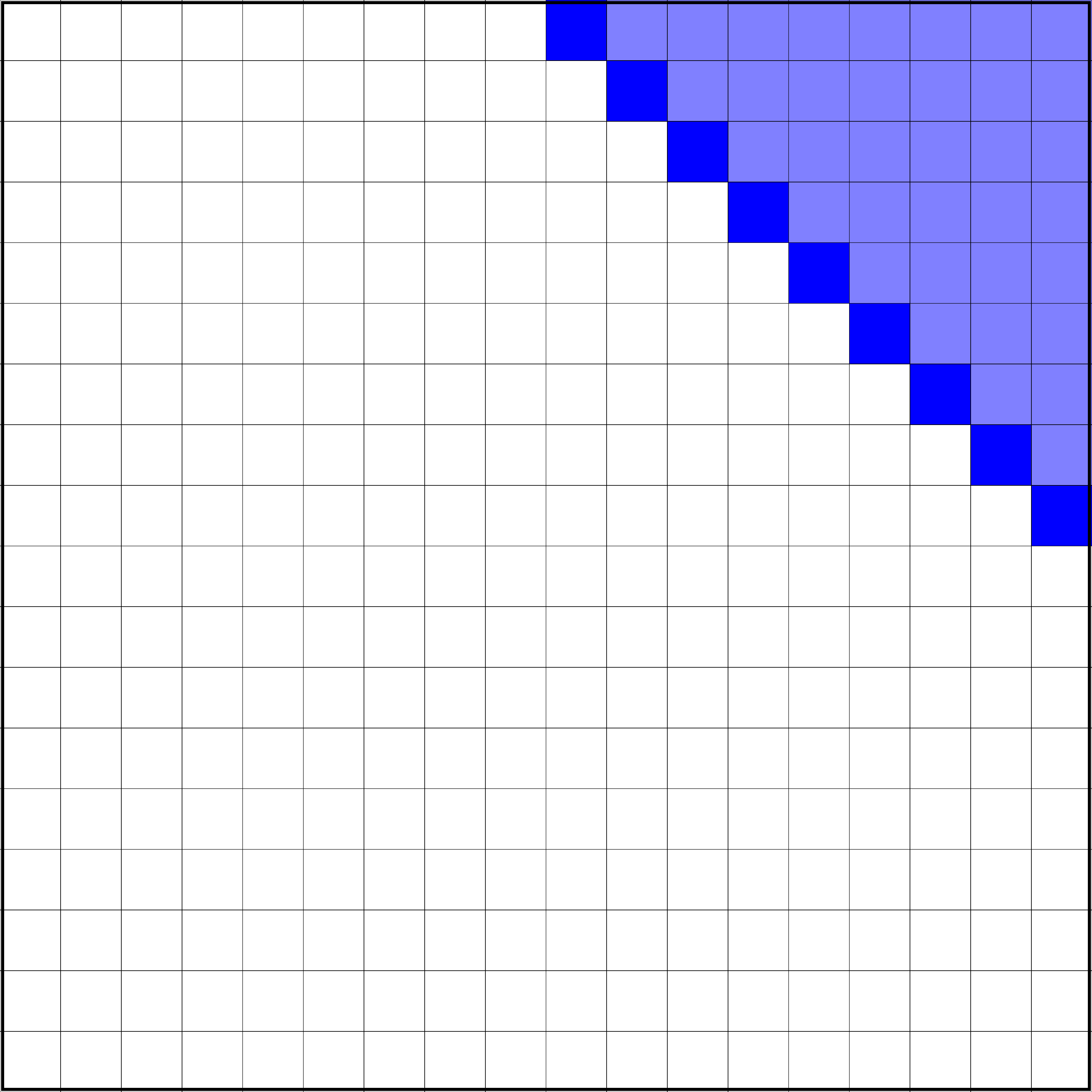}
\put(-80,50){$1 + \epsilon$}
\put(-20,95){$\xi$}
\tiny
\put(-90,118){$M$}
\put(-40,118){$n-M$}
  \hspace{1in}
   \includegraphics[scale=.08]{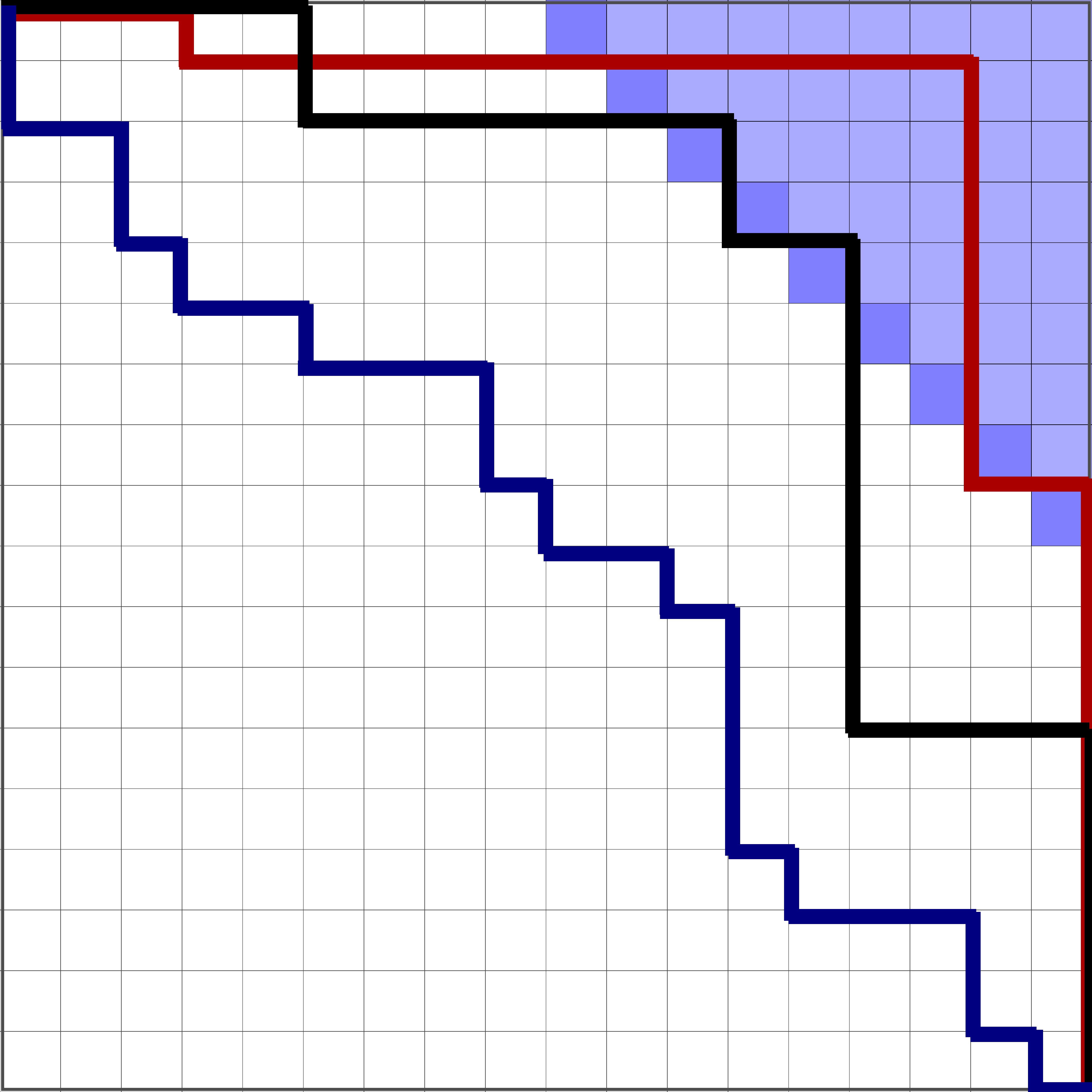}
\caption[ Fluctuating bias with exponential mixing time]{(a) Fluctuating bias with exponential mixing time. (b) Staircase walks in $S_1, S_2,$ and $S_3$.}
\label{fluctuatingbias}
\end{figure}

\vspace{.2in}
\noindent {\em Proof of Theorem \ref{SlowSimple}.} \ 
%\begin{proof}
We will define values $\{\lambda_{x,y}\}_{(x,y)\in R}$ such that the conductance of the Markov chain $\mfluct$ is small.  Let $M= n-\sqrt{n}.$  For all $(x,y)$ such that $x+y\leq n+M$, define $\lambda_{x,y}=1+\epsilon$, where $\epsilon = \frac{1}{4n}$.  For all remaining $(x,y),$ let $\lambda_{x,y}=\xi,$ where $\xi>1$ will be defined later.
  For a staircase walk $\sigma$ consisting of a sequence of steps $\sigma_{i}\in\{\pm 1\}$, define the \emph{height of $\sigma_i$} as $\sum_{j\leq i} \sigma_{j}$, and let $\max(\sigma) $ be the maximum height of $\sigma_i$ over all $1\leq i\leq 2n$.
Let $S_1$ be the set of configurations $\sigma$ such that $\max(\sigma)< n+M$, $S_2$ the set of configurations such that $\max(\sigma)= n+M$, and $S_3$ the set of configurations such that $\max(\sigma)> n+M$.

First we will show that $\pi(S_2)$ is exponentially smaller than $\pi(S_1)$ for all values of $\xi$.  Let $A(\sigma)$ be the number of tiles contained in $\sigma$.  Then

$$\pi(S_2)=\frac{1}{Z}\sum_{\sigma\in S_2}(1+\epsilon)^{A(\sigma)}\leq \frac{(1+\epsilon)^{n^2-n/2}|S_2|}{Z},$$
since there are at most $n^2-(n-M)^2/2=n^2-n/2$ tiles with weight $1+\epsilon$ in $\sigma$.  By the definition of $\epsilon$, we have $(1+\epsilon)^{n^2-n/2} = (1+  \frac{1}{4n})^{n^2-n/2}  \leq  e^{n/4 - 1/8}. $  Hence $\pi(S_2)\leq e^{n/4-1/8}|S_2|/Z$.
%%%%%%%%%%%%%%%%

Next we will bound $|S_2\cup S_3|$, which in turn provides an upper bound on $|S_2|$.  The unbiased Markov chain is equivalent to a simple random walk $W_{2n}=X_1+X_2+\cdots +X_{2n}=0$, where $X_i\in \{+1,-1\}$ and where a $+1$ represents a step to the right and a $-1$ represents a step down.  We call this random walk {\it tethered} since it is required to end at $0$ after $2n$ steps.  Compare walk $W_{2n}$ with the untethered simple random walk $W_{2n}'=X_1'+X'_2+\ldots + X'_{2n}$.
\begin{eqnarray}
P\left(\max_{1\leq t\leq 2n} W_t\geq M\right)&=&P\left(\max_{1\leq t\leq 2n} W'_t\geq M ~|~ W'_{2n}=0\right)\nonumber\\
&=&\frac{ P\left(\max_{1\leq t\leq 2n} W'_t\geq M\right)}{ P(W'_{2n}=0) }\nonumber\\
&=&  \frac{2^{2n}}{\binom{2n}{n}}P\left(\max_{1\leq t\leq 2n} W'_t\geq M\right)\nonumber\\
&\leq&\frac{e^2\sqrt{n}}{2\sqrt{ \pi}}~~P\left(\max_{1\leq t\leq 2n} W'_t\geq M\right)\label{eq:sterling}\\
&<&3\sqrt{n}~~P\left(\max_{1\leq t\leq 2n} W'_t\geq M\right),\nonumber
\end{eqnarray}
where~\ref{eq:sterling} follows from Sterling's formula. Since the $\{X'_i\}$ are independent, we can use Chernoff bounds to see that

$$P\left(\max_{1\leq t\leq 2n}W'_t\geq M\right) \ \leq \  2nP(W'_{2n}\geq M) \ \leq \  2n e^{\frac{-M^2}{2n}}.$$
Notice that $M^2/(2n) = (n-\sqrt{n})^2/(2n)= (\sqrt{n}-1)^2/2\geq n/3 $ for $n\geq 4$.
Together these show that $P\left(\max_{1\leq t\leq 2n} W_t\geq M\right)\leq 6 n^{3/2}e^{-n/3}$.  In particular, 
$$|S_2\cup S_3|\leq 6 \binom{2n}{n}n^{3/2} e^{-n/3}.$$
Therefore we have
\begin{eqnarray*}
\frac{\pi(S_2)}{\pi(S_1)}&\leq &\frac{\frac{1}{Z} e^{n/4-1/8}|S_2\cup S_3|}{\frac{1}{Z} \left(\binom{2n}{n} - |S_2\cup S_3|\right)}\\
&\leq&  e^{n/4-1/8} \left(\frac{\binom{2n}{n}}{|S_2\cup S_3|} - 1\right)^{-1}\\
&\leq&  e^{n/4-1/8}\left(\frac{\binom{2n}{n}}{ 6\binom{2n}{n}n^{3/2} e^{-n/3}} - 1\right)^{-1}\\
&\leq& \frac{6e^{n/4-1/8}n^{3/2}}{e^{n/3} - 6n^{3/2}} \ < \ e^{-n/24},
\end{eqnarray*}
for large enough $n$.  Therefore, $\pi(S_2)$ is exponentially smaller than $\pi(S_1)$ for every value of $\xi$.

Our next goal is to show that there exists a value of $\xi$ for which $\pi(S_3)=e\pi(S_1)$, which will imply that $\pi(S_2)$ is also exponentially smaller than $\pi(S_3)$, and hence the set $S_2$ forms a bad cut, regardless of which state the Markov chain begins in.  To find this value of $\xi$, we will rely on the continuity of the function $g(\xi) = Z\pi(S_3) - Z\pi(S_1)$ with respect to $\xi$.  Notice that $Z\pi(S_1)$ is constant with respect to $\xi$ and 
$$Z\pi(S_3) = \sum_{\sigma\in S_3 } (1+\epsilon)^{A_b(\sigma)} \xi^{A_a(\sigma)},$$
(where $A_b(\sigma)$ is the number of tiles below the diagonal $M$ in $\sigma$ and $A_a(\sigma)$ is the number of tiles above the diagonal $M$ in $\sigma$)
 is just a polynomial in $\xi$.  Therefore $Z\pi(S_3)$ is continuous in $\xi$ and hence $g(\xi)$
%$$g(\xi) = \sum_{\sigma\in S_3} (1+\epsilon)^{A_b(\sigma)} \xi^{A_a(\sigma)} - \sum_{\sigma\in S_1} (1+\epsilon)^{A_b(\sigma)}$$
is also continuous with respect to $\xi$.  
We will show that $g(1+\epsilon) <0$ and  $g(4e^{2.5})>0$, so by continuity we will conclude that there exists a value
  of $\xi$ satisfying $1+\epsilon < \xi < 4e^{2.5}$ for which $g(\xi)=0$ and
  $Z\pi(S_3)=eZ\pi(S_1)$.  Clearly this implies that for this choice of $\xi$,
  $\pi(S_3)=e\pi(S_1)$, as desired. 
  % To obtain the corresponding value of
  %$\delta$, we notice that $\delta = 1/(\xi+1)$.  In particular, $\delta$ is a
  %constant satisfying $\frac{1}{33}<\delta<\frac{1}{\gamma+1}< \frac{1}{2}.$

First, we will show that when $\xi = (1+\epsilon)$, $Z\pi(S_3)< eZ\pi(S_1)$, so $g(1+\epsilon)<0$.  
This is easy to see, as
\begin{equation*}
\pi(S_3)~\leq~ \frac{1}{Z}|S_3| \gamma^{n^2} ~\leq ~\frac{1}{Z}|S_3| e~\leq ~\frac{1}{Z}|S_1| e~
\leq~ e\pi(S_1).
\end{equation*}

Next we will show that $g(4e^{2.5})>0$.
First we notice that since the maximal tiling is in $S_3$, 
$$\Pi(S_3)\geq {Z}^{-1} (1+\epsilon)^{n^2-\frac{(n-M)^2}{2}}\xi^{\frac{(n-M)^2}{2}}.$$
Also, 
$$\Pi(S_1)={Z}^{-1}\sum_{\sigma\in S_1}(1+\epsilon)^{A_a(\sigma)} \ < \ Z^{-1}\binom{2n}{n}(1+\epsilon)^{n^2-\frac{(n-M)^2}{2}}.$$
  Therefore

$$e\pi(S_1)/\pi(S_3)< \frac{e\binom{2n}{n}}{\xi^{\frac{(n-M)^2}{2}}}
\leq {e(2e)^n}{\xi^{-n/2}} \leq 1 $$
for $n\geq 4$, since $\xi=4e^{2.5}$.  Hence $g(4e^{2.5}) = Z\pi(S_3) - eZ\pi(S_1)> 0,$ as desired.

 Finally, we may analyze the conductance.  
 $$\Phi \leq \phi_{S_1} \leq
\frac{1}{\pi(S_1)}  \sum_{x\in S_1}\pi(x)\sum_{y\in S_2} P(x,y)\leq  
\frac{1}{\pi(S_1)} \sum_{x\in
  S_1}{\pi(x)}\pi(S_2)\leq 
  e^{-n/24}.$$

  \noindent Hence, by Theorem~\ref{conductance}, the mixing time of
  $\mmon$ satisfies
  $$\tau\geq ({4e^{-n/24}})^{-1} - 1/2 \geq e^{n/24}/4-1/2.$$
%\end{proof}
\hfill \square
\vskip.2in

\section{Conclusions}
\label{Conclusion}

In this paper, we showed that the DNA-based self-assembly is efficient for any uniform bias in two dimensions and for large enough bias in higher dimensions.  We also gave the first analysis of self-assembly with fluctuating bias, showing that as long as the biases do not differ by too much, the assembly is still rapid.
We have made several contributions to the study of biased card shuffling as well.  The bound in Theorem~\ref{2DSimple} on the mixing time of $\mmon$ when $d=2$ yields a simpler proof that the nearest-neighbor transposition chain on biased permutations is rapidly mixing, using the bijection from Benjamini et al.~\cite{BBHM05}.  In fact, we achieved the same optimal bounds on the mixing time.  Recently our improved bounds on the mixing time of $\mmon$ for rectangular regions were used to show a tighter bound on the mixing time of the nearest-neighbor transposition chain for a generalization of biased permutations arising in the context of self-organizing lists~\cite{bmrs}.

The techniques from Section~\ref{generalization} can be extended to other applications, such as biased 3-colorings. There is a well-known bijection between $3$-colorings of $\Z^2$ and sets of monotonic, {\it edge}-disjoint paths (see, e.g. \cite{lrs}).  The construction generalizes to arbitrary dimension as well, forming ($d-1$)-dimensional monotonic surfaces that are {\it face}-disjoint.  There is a Markov chain $\mcol$ arising in the context of asynchronous cellular automata which samples biased 3-colorings.  Our technique can be used to show that as long as the bias satisfies $\lambda\geq 4d^2$, the mixing time of $\cM_{Col}$ satisfies $\tau(\varepsilon)=O\left(n^2 \ln(\varepsilon^{-1})\right).$  We expect that there may be other situations in which this technique could be useful as well.

\bibliographystyle{plain}
\bibliography{bib}
\end{document}